\documentclass[manuscript]{acmart}
\acmJournal{TQC}
\usepackage[T1]{fontenc}
\usepackage{graphicx}

\usepackage{amsmath, hyperref}
\usepackage{amsthm}
\usepackage{enumitem}
\usepackage{dsfont}
\usepackage{pdfpages}
\usepackage{esint}
\usepackage{xfrac}
\usepackage{comment}
\usepackage{complexity}
\usepackage[dvipsnames]{xcolor}
\usepackage{xspace}
\usepackage{physics}
\usepackage{braket}
\usepackage{algorithm}
\usepackage[noend]{algpseudocode}
\usepackage{cleveref}
\usepackage{tikz}
\usepackage{mleftright}

\usepackage{caption}
\usepackage{subcaption}

\usepackage{parskip}

\usepackage[disable]{todonotes}

\usepackage{thmtools}
\usepackage{thm-restate}

\usepackage{tikz}
\usetikzlibrary{quantikz2}

\newtheorem{note}{Note}

\DeclareMathOperator*{\argmin}{arg\,min}

\usetikzlibrary{automata}

\newcommand\defmath[2]{\newcommand#1{\ensuremath{#2}\xspace}}
\newcommand\concept[1]{\textit{#1}}

\undef\co
\newcommand{\co}{\mathds{C}}

\defmath{\bool}{\mathbb{B}}
\defmath{\complex}{\mathbb{C}}
\defmath{\integers}{\mathbb{Z}}

\definecolor{rowcolor}{RGB}{223, 213, 242}
\definecolor{colcolor}{RGB}{196, 181, 240}

\definecolor{origgreen}{RGB}{196, 181, 240}
\definecolor{origblue}{RGB}{196, 181, 240}
\definecolor{origpurple}{RGB}{196, 181, 240}
\definecolor{origred}{RGB}{196, 181, 240}

\tikzstyle{leaf}=[draw, rectangle,minimum size=4.mm, inner sep=3pt]
\tikzstyle{var}=[circle,draw=black!70,solid,thick,minimum size=6mm]
\tikzstyle{bdd}=[regular polygon, regular polygon sides=3, draw=black!70,solid,thick,inner sep=0.5mm]
\tikzstyle{n}=[->,loosely dashed,thick]
\tikzstyle{p}=[->,solid,thick]
\tikzstyle{b}=[->,densely dashdotted,ultra thick]
\tikzstyle{e0}[0]=[dotted,thick,bend right=#1]
\tikzstyle{e1}[0]=[solid, bend left =#1]

\tikzstyle{lbl}=[draw,fill=white,inner sep=2pt, minimum size=0cm,line width=.5pt]

\renewcommand\index{\textsf{idx}\xspace}

\newcommand\rootlabel{\textsf{RootLabel}\xspace}
\newcommand\unique{\textsc{Unique}\xspace}
\newcommand\cache{\textsc{Cache}\xspace}

\newcommand\Edge{\textsc{Edge}\xspace}
\newcommand\Node{\textsc{Node}\xspace}

\newcommand\Pauli{\textsc{Pauli}\xspace}
\newcommand\pauli{\Pauli}
\newcommand\makeedge{\textsc{MakeEdge}\xspace}

\newcommand\follow[2]{\ensuremath{\textsc{follow}_{#1}(#2)}}
\newcommand\getautomorphisms{\textsc{GetStabilizerGenSet}}

\newcommand\low[1]{\ensuremath{\textsf{low}(#1)}}
\newcommand\high[1]{\ensuremath{\textsf{high}(#1)}}

\newcommand{\limdd}{\textsf{LIMDD}}
\newcommand\rootlim{B_{\textnormal{root}}}

\newcommand\highlim{B_{\textnormal{high}}}

\newcommand{\id}{\mathds{I}}

\newcommand{\beforeq}{\preccurlyeq}

\tikzstyle{leaf}=[draw, rectangle,minimum size=4.mm, inner sep=3pt]

\newlength{\pgfcalcparm}
\newlength{\pgfcalcparmm}

\DeclareRobustCommand{\leafedge}[2][]{%
  \pgftext{\settowidth{\global\pgfcalcparm}{\scriptsize $\,\,\,#2\,\,\,$}}%
  \raisebox{-.8mm}{%
  \tikz{%
    \node[inner sep=0pt] (x){};%
    \node[right=\pgfcalcparm of x,leaf](v){\scriptsize 1};%
    \draw (x) to node[above,pos=.5]{\scriptsize $\,#2\,\,$} (v);%
  }%
  }%
}

\DeclareRobustCommand{\ledge}[3][]{%
  \pgftext{\settowidth{\global\pgfcalcparm}{\scriptsize $\,\,\,#2\,\,\,$}}%
  \raisebox{-.8mm}{%
  \tikz{%
    \node[inner sep=0pt] (x){$#1\,\,$};%
    \node[state,inner sep=0pt,minimum size=10pt,right=\pgfcalcparm of x](v){\scriptsize $#3$};%
    \draw (x) to node[above,pos=.5]{\scriptsize $\,#2\,\,$} (v);%
  }%
  }%
}

\DeclareRobustCommand{\lnode}[5][]{%
    \pgftext{\settowidth{\global\pgfcalcparm}{\scriptsize $\,\,\,#2\,\,\,$}}%
    \pgftext{\settowidth{\global\pgfcalcparmm}{\scriptsize $\,\,\,#4\,\,\,$}}%
  ~\raisebox{-.mm}{%
  \tikz{%
  \vspace{-1mm}%
    \node[state,inner sep=0pt,minimum size=10pt] (v){\scriptsize $#1$};%
    \node[state,inner sep=0pt,minimum size=10pt,left=\pgfcalcparm of v](v0){\scriptsize $#3$};%
    \draw[dotted] (v) to node[above,pos=.45]{\scriptsize $#2$} (v0);%
    \node[state,inner sep=0pt,minimum size=10pt,right=\pgfcalcparmm of v](v1){\scriptsize $#5$};%
    \draw (v) to node[above,pos=.45]{\scriptsize $#4$} (v1);%
  }%
  }%
}

\def\qisq{\mathfrak{F}}

\newcommand{\redacted}[1]{**REDACTED**}

\def\cktstates{\mathcal{S}}
\def\algebraic{{\texttt{algebraic}}}
\def\float{{\texttt{float}}}

\tikzset{every picture/.style={->,thick}}

\title{
Exact quantum decision diagrams with scaling guarantees for Clifford+$T$ circuits and beyond
}

\author{Arend-Jan Quist}
\orcid{0000-0002-6501-2112}
\affiliation{%
\institution{Leiden University}
\city{Leiden}
\country{The Netherlands}}
\author{Tim Coopmans}
\orcid{0000-0002-9780-0949}
\affiliation{%
\institution{Delft University of Technology}
\city{Delft}
\country{The Netherlands}}
\author{Alfons Laarman}
\orcid{0000-0002-2433-4174}
\affiliation{%
\institution{Leiden University}
\city{Leiden}
\country{The Netherlands}}
\email{{a.quist,a.w.laarman}@liacs.leidenuniv.nl, t.j.coopmans@tudelft.nl}

\begin{document}

\begin{abstract}
A decision diagram (DD) is a graph-like data structure 
for homomorphic compression of Boolean and pseudo-Boolean functions.
Over the past decades, decision diagrams have been successfully applied to verification, linear algebra,
stochastic reasoning, and quantum circuit analysis.
Floating-point errors have, however, significantly slowed down practical implementations of real- and complex-valued decision diagrams.
In the context of quantum computing, attempts to mitigate this numerical instability have thus far lacked theoretical scaling guarantees and have had only limited success in practice.
Here, we focus on the analysis of quantum circuits consisting of Clifford gates and $T$ gates (a common universal gate set). 
We first hand-craft an algebraic representation for complex numbers, which replace the floating point coefficients in a decision diagram.
Then, we prove that the sizes of these algebraic representations are linearly bounded in the number of $T$ gates and qubits, and constant in the number of Clifford gates.
Furthermore, we prove that both the runtime and the number of nodes of decision diagrams are upper bounded as $2^t \cdot \poly(g, n)$, where $t$ ($g$) is the number of $t$ gates (Clifford gates) and $n$ the number of qubits.
Our proofs are based on a $T$-count dependent characterization of the density matrix entries of quantum states produced by circuits with Clifford+$T$ gates, and uncover a connection between a quantum state's stabilizer nullity and its decision diagram width.
With an open source implementation, we demonstrate that our exact method resolves the inaccuracies occurring in floating-point-based counterparts and can outperform them due to lower node counts.
Our contributions are, to the best of our knowledge, the first scaling guarantees on the runtime of (exact) quantum decision diagram simulation for a universal gate set.
\end{abstract}

\maketitle

\section{Introduction}

A decision diagram~\cite{bryant86} (DD) is a data structure designed originally for the succinct representation and fast manipulation of Boolean functions $f \colon \{0, 1\}^n \to \{0,1\}$, or, equivalently, of binary vectors $\vec{v}_f$ of length $2^n$.
By representing the function as a directed acyclic graph (DAG) and merging isomorphic sub-graphs, DDs avoid storing equal sub-functions that arise from different partial assignments to the Boolean variables.
For example, if the function $f$ satisfies equality between the two subfunctions $f(1, \dots, 1,  x_k, \dots, x_n) = f(0, \dots, 0,  x_k, \dots, x_n)$, then we may save $2^{n-k+1}$ bits for storing $f$ by only storing this subfunction once. A DD implements this by merging all nodes representing the same subfunction.%

By putting weights on the decision diagram leaves~\cite{fujita1997multi} (Multi-Terminal Binary Decision Diagram; MTBDD) or on its edges~\cite{lai1994evbdd,tafertshofer1994factored,tafertshofer1997factored} (Edge-Valued Binary Decision Diagram; EVDD), decision diagram can represent pseudo-Boolean functions $f: \{0, 1\}^n \mapsto V$ where, for instance, $V = \mathbb{R}$ or $\mathbb{C}$, or, equivalently, real- and complex-valued vectors.
Edge-valued decision diagrams merge sub-graphs that represent \emph{equivalent} instead of fully equal sub-functions, e.g. when there exists a constant $\lambda \in \mathbb{C}$ such that $f(1, \dots, 1,  x_k, \dots, x_n) = \lambda \cdot f(0, \dots, 0,  x_k, \dots, x_n)$ for all $x_k, x_{k+1}\, \dots, x_n \in \{0, 1\}$.
This can yield exponentially better compression compared to a diagram that only stores values in the leaves and not on the edges~\cite{fargier2014knowledge,fargier2013semiring}.
Going beyond vectors, decision diagrams are also able to represent matrices~\cite{fujita1997multi,mcmillan,bryant1995verification} or higher-dimensional tensors~\cite{hong2022tensor}, enabling matrix-vector multiplication and other linear-algebra tasks.

Since exponentially-sized, yet highly structured complex vectors and matrices are at the core of quantum information, decision diagrams are a natural tool for reasoning over quantum systems.
Indeed, tailored algorithms for decision diagrams have been developed in the past for reasoning over quantum circuits and have proven successful for various tasks, for example, in classical simulation of quantum circuits~\cite{miller2006qmdd,vinkhuijzen2023efficient}, equivalence-checking of quantum circuits~\cite{burgholzer2020improved}, and synthesis~\cite{niemann2020advanced}.
Such applications represent the quantum state vector as a decision diagram: for example, for (strong) quantum-circuit simulation in particular, a decision diagram is first created for the input state $\ket{\phi_0}=\ket{0}^{\otimes n}$, after which by iteration over the gates in the circuit, each time an algorithm is called to update the DD representation of the quantum state $\ket{\phi_i}$ and replace it by $\ket{\phi_{i+1}}$. See \Cref{fig:circuit_example} for an example. A separate algorithm simulates measurement, extracts information, and updates the DD.
Other successful applications of continuous decision diagrams include probabilistic model checking and stochastic reasoning~\cite{sanner2005AffineADDs}, as well as linear algebra with sparse and/or highly structured matrices in other contexts~\cite{fujita1997multi}.

\begin{figure}[tbh]
    \centering
    \input{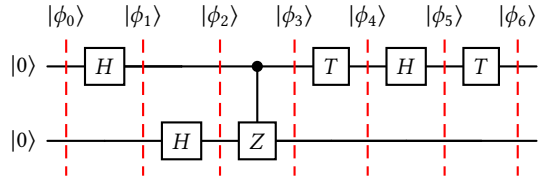}
    \caption{An example of a quantum circuit with 2 qubits and 6 gates. Simulating this circuit consists of first creating the input state $\ket{\phi_0}=\ket{0}^{\otimes 2}$. Then, iteratively, the states $\ket{\phi_i}$ are updated according to the next gate and replaced by the state $\ket{\phi_{i+1}}$. \\
    \Cref{fig:leadingExample} shows the simulation of this circuit with decision diagrams, where the states $\ket{\phi_i}$ are represented by an EVDD and LIMDD.}
    \label{fig:circuit_example}
\end{figure}

In this work, we focus on the application of two types of decision diagrams to quantum circuits, namely the well-developed Edge-Valued Decision Diagram~\cite{lai1994evbdd,tafertshofer1994factored,tafertshofer1997factored} (EVDD; see also the Quantum Multi-Valued Decision Diagram~\cite{miller2006qmdd}) and its successor Local Invertible Map Decision Diagram~\cite{vinkhuijzen2023limdd} (LIMDD).
LIMDD was proposed to remedy a limitation of EVDD: It cannot efficiently simulate Clifford circuits, despite the fact that this problem is classically tractable using the so-called \concept{stabilizer formalism}~\cite{aaronson2004improved}.
LIMDD effectively combines the stabilizer and decision diagram formalisms by merging nodes if they are equal up to a constant factor $\lambda$ (as EVDD), \emph{and} up to local Pauli operators (i.e., $X, Y, Z$) that capture rotational symmetries in the stabilizer formalism.
As a consequence, LIMDD can polynomially simulate both Clifford circuits and EVDD, and is in the best case exponentially more succinct than their union~\cite{vinkhuijzen2024a}. %
EVDDs have had some success in quantum-circuit analysis~\cite{zulehner2018advanced}, whereas the theoretically proven exponential improvement of LIMDDs over EVDDs has not yet materialized in implementation.

In contrast to the binary case, software implementations of decision diagrams for pseudo-Boolean functions, such as EVDD and LIMDD, suffer from finite floating-point precision.
The reason for this is that merging two DD nodes typically requires the function (or vector) that they represent parts to be \emph{exactly} equivalent.
In implementations, however, due to small errors arising from finite-precision arithmetic on real and complex numbers, vector parts are no longer exactly equivalent, forcing the decision diagram to either duplicate nodes that represent equal subfunctions or to consider epsilon-equivalence, risking unjustified merges due to rapidly propagating errors~\cite{brand2025numerical}.
The accumulation of finite-accuracy errors after many decision-diagram operations has indeed proven disastrous in practice, undoing some of the exponential improvements that DDs achieve with respect to, e.g., standard matrix-vector operations or between the various DD variants~\cite{sanner2005AffineADDs,brand2025numerical}{do we actually have a reference for this claim?}.\footnote{Scott Sanner: ``Since 2005 this problem [blow ups due to numerical errors] vexes me. It's ruined a year of a PhD's life without a solution." (\url{https://www.youtube.com/watch?v=Xi7yOXi--Kc&list=WL})} 
See \autoref{sec:motivation} for a motivating example and existing work.

Moreover, the accumulation of finite-accuracy errors during simulation may also lead to a final quantum state that is substantially different from the actual quantum state, rendering the simulation useless~\cite{niemann2020overcoming,brand2025numerical}. Since it is not clear from the simulation result itself whether it was computed correctly, the reliability of quantum circuit simulation with floating-point-based DD implementations is limited. 

Recent progress has increased the understanding of numerical accuracy in decision diagrams. For MTBDDs, analysis demonstrates that the bound on the worst-case error is reasonable~\cite{brand2025numerical}. For non-canonical data structures, admitting only real nonnegative weights, a provably efficient approximation of complex numbers exists bounding the error to arbitrary precision~\cite{bryant2025numerical}. As quantum computing requires complex weights, or at least negative weights~\cite{mei2024eq,Mei_2024}, and MTBDD has been outperformed by EVDD and LIMDD in both theory and practice, these previous works are unsatisfactory for quantum circuit analysis.

In this paper, we therefore adopt an exact algebraic representation of the weights occurring in quantum circuits. While it is known from Kitaev's theorem that there are finite universal gate sets for quantum computing, and hence that all the necessary weights can be captured with an extended ring~\cite{giles2013exact}, little is known about the complexity of these representations in terms of the original quantum circuit.
This leads to our first challenge.

\textit{
\textbf{Challenge 1:} What is the complexity of solving finite-accuracy errors in decision diagrams with an algebraic representation?
}

A slightly related problem ---in the sense that it also limits the applicability of decision diagrams--- is the lack of scaling guarantees for the size of decision diagrams for quantum circuit simulation. In practice we see that the number of nodes in decision diagrams during simulation can vary hugely between different (types of) quantum circuits. In the literature, there is only a limited understanding which circuits require large or small decision diagrams to simulate.

The field of knowledge compilation~\cite{darwiche2002knowledge} gives some bounds by studying the effect on the size of a DD when applying operations (such as gate updates and measurements). %
Here, DD operations are studied in isolation~\cite{vinkhuijzen2024a,fargier2014knowledge}, leading to statements of the form: `Gate $G$ can be applied to a $EV/LIM$DD in time $poly(n)$, increasing the size of the DD by a factor at most $poly(n)$.' However, when simulating a circuit consisting of multiple gates, this gives exponential bounds on the DD representing the final quantum state, even if every gate only doubles the DD size. 
We intend to improve on this by pursuing the following challenge.

\textit{
\textbf{Challenge 2:} Can we improve existing bounds on DD size for quantum circuit simulation by extending single gate analysis to an entire quantum circuit?
}

Since the problem of simulating quantum circuits is not expected to be tractable (let alone their analysis), the best we can hope for is fixed-parameter tractability (fpt). Indeed, it has been shown before that simulating the Clifford + $T$ gate set is fpt in the number of $T$ gates~\cite{bravyi2019simulation}.
In this paper, 
we also consider simulation of $n$-qubit circuits with an unbounded number of Clifford gates but a fixed number $t$ of $T$ gates, which we refer to as the ``$T$ count.'' Resolving both challenge 1 and 2 yields a novel fpt-method for quantum circuit simulation that is also implementable with exactness guarantees.

\subsection{Contributions}
Our contributions address the challenges described above and demonstrate the practicality of the proposed approach:

\begin{enumerate}
\item In \Cref{sec:coeff-bounds}, we resolve the numerical inaccuracy issue by providing an algebraic representation of the complex numbers in the edge values of EVDDs and LIMDDs. The size of the algebraic representation scales linearly in $t$ and $n$ when simulating an $n$-qubit Clifford+$t\times T$ circuits, i.e., independently of the number of Clifford gates.
This enables \emph{exact} simulation of such circuits with algebraic expressions instead of floating point numbers.
This goes beyond the work of~\cite{giles2013exact}, which characterizes the \textit{entries} of a \textit{unitary} representing a Clifford+$t \times T$ circuit only for the limit case $t\to\infty$, while our work characterizes \textit{DD edge values} of a \textit{quantum state} produced by a Clifford+$t \times T$ circuit for any (finite) value of $t$ asymptotically.

Our proof consists of two steps.
First, we prove that the entries of the density-matrix representing the state computed by an $n$-qubit Clifford+$t\times T$ circuit multiplied by $2^n$ lie in a set $Q_{n,t} \subseteq \mathbb{C}$, where each element of $Q_{n,t}$ can be written using only linearly many bits in $t$ and $n$.
Our next ingredient is the insight that a characterization of density-matrix entries translates to a characterization of decision-diagram edge labels when the state vector is represented as either EVDD or LIMDD, while preserving the space usage linear in $t$ and $n$.

Finally, we extend this result
to quantum circuit \textit{simulation} by analyzing DD operations.

\item In \Cref{sec:numnodes_bound}, our second contribution is a bound on the number of decision diagram nodes 
when performing quantum-circuit simulation of Clifford+$t\times T$ circuits, for both EVDD and LIMDD, assuming that simulation is done algebraically (which we have proven can be done efficiently in our first contribution).
Specifically, we prove that the output state of a Clifford+$t\times T$ circuit is represented by a LIMDD of $O(2^t)$, while proving a weaker statement for EVDD.

Our proof is based on the observation that the decision diagram width is bounded as function of the (local) stabilizer nullity~\cite{beverland2020lower} of the state it represents which is unchanged by Clifford gates while each $T$-gate can only increment the nullity by one.
We also extend our results to the universal Clifford + Toffoli gate set.

\item 
An open-source implementation, extending the decision diagram package Q-Sylvan~\cite{brand2025q} with algebraic representation from \autoref{sec:coeff-bounds}, shows that the new exact method outperforms the floating-point EVDD implementation and can solve cases where the latter is impractical (\autoref{sec:experiments}).
Our results thus show the potential of DD-based quantum-circuit analysis, and provide new proof techniques for overcoming size blow-up for other continuous-valued decision diagram applications.
\end{enumerate}

Combining the two first results, we prove that \emph{exact} simulation with LIMDDs is possible with space polynomial in the number of qubits and Clifford gates, and exponential in the number of $T$ gates, i.e., that LIMDD simulation is \emph{fixed-parameter tractable}~\cite{cygan2015parameterized} in the $T$-count. As decision diagram algorithms are inherently recursive, producing no nodes that are not part of the final output~\cite{wegener2000branching}, node number bounds automatically translate into runtime bounds.
Therefore, we can rephrase and summarize our theoretical result in the following corollary:
\begin{corollary}
Exact simulation of a Clifford+$T$ circuit with LIMDD requires at most exponential time in the number of $T$ gates and polynomial time in the number of Clifford gates and the number of qubits.\\
For EVDD, it is exponential in the minimum of $H$ gates and the number of $CZ$ and $T$ gates, and polynomial in the number of other single qubit Clifford gates and the number of qubits.
\end{corollary}

\section{Preliminaries \label{sec:prelims}}

\subsection{Quantum computing}
\label{sec:qc}

The state of a set of $n$ quantum bits, the units of quantum computing~\cite{nielsen2010quantum}, is described by a normalized complex vector of $2^n$ numbers, typically denoted $\ket{v}$ in Dirac notation:
\begin{equation*}
    \ket{v} = \begin{pmatrix}
        v_{0}&
        v_{1}&
        v_{2}&
        \dots&
        v_{2^n-1}
    \end{pmatrix}^{\top}, \qquad \forall j: v_{j}\in\co, \qquad \text{such that }\sum_{j=0}^{2^n-1}|v_j|^2=1.
\end{equation*}
Alternatively, we can write $\vec{v} = \sum_{j \in \{0, 1\}^n} v_j \ket{j}$ where we interpret the subscript $j$ in $v_j$ as integer in binary notation. Here, the states $\ket{00\dots0}, \ket{00\dots 01}, \dots, \ket{11\dots1}$ are the $n$-qubit computational-basis states, defined as $\ket{j}_k = 1$ if $k=j$ and $0$ otherwise, for $k \in \{0, 1\}^n$.

Quantum gates transform one $n$-qubit quantum state to another $n$-qubit quantum state.
An $n$-qubit gate is described by an $2^n \times 2^n$ unitary matrix $U$ and the update is performed by matrix-vector multiplication: $\ket{v} \mapsto U\cdot \ket{v}$.
Common single-qubit gate are the Pauli gates:
\begin{equation*}
    I=\begin{pmatrix}
        1&0\\
        0&1
    \end{pmatrix},
    \quad
    X=\begin{pmatrix}
        0&1\\
        1&0
    \end{pmatrix},
    \quad
    Y=\begin{pmatrix}
        0&-i\\
        i&0
    \end{pmatrix},
    \quad
    Z=\begin{pmatrix}
        1&0\\
        0&-1
    \end{pmatrix}.
\end{equation*}

The state of the disjoint union of two registers of $n_A$ resp. $n_B$ quantum bits, given by quantum states $\ket{\phi_A}, \ket{\phi_B}$, is given by the $n_A + n_B$-qubit state $\ket{\phi_A} \otimes \ket{\phi_B}$ where $\otimes$ denotes Kronecker product.
A Kronecker product of two matrices $A$ and $B$ is defined by 
\begin{equation*}
    A\otimes B = 
    \begin{pmatrix}
        a_{11}& a_{12}&\dots& a_{1l}\\
        a_{21}& a_{22}&\dots& a_{2l}\\
        \vdots&\vdots &&\vdots\\
        a_{k1}& a_{k2}&\dots& a_{kl}
    \end{pmatrix}\otimes B
    = \begin{pmatrix}
        a_{11}B& a_{12}B&\dots& a_{1l}B\\
        a_{21}B& a_{22}B&\dots& a_{2l}B\\
        \vdots&\vdots &&\vdots\\
        a_{k1}B& a_{k2}B&\dots& a_{kl}B
    \end{pmatrix}
\end{equation*}
where $a_{ij}B$ indicates the submatrix $B$ multiplied by the scalar $a_{ij}$.

Similarly, applying $n_A$-qubit unitary $U_A$ in parallel with $n_B$-qubit unitary $U_B$ is the same as the application of the unitary $U_A \otimes U_B$.
The identity gate $I$ is used for the application of a gate to a subregister only, e.g., $I\otimes H\otimes I\otimes\dots\otimes I$ is the application of $H$ to the second qubit and $I\otimes I\otimes T\otimes I\otimes\dots\otimes I$ is the application of $T$ to the third qubit.
We will denote $U_k$ for applying the gate $U$ to qubit $k$.
Sequential composition of unitaries $U_A, U_B$ corresponds to the matrix product $U_A \cdot U_B$.

The closure of the gate set $\{H, S=T^2, CZ\}$
under parallel and sequential composition is called the `Clifford group'.
Together with the $T$ gate, the set becomes a universal gate set~\cite{dawson2005solovaykitaev}: any unitary can be approximated arbitrarily close by circuits using only gates from this set. In this work, we mainly consider Clifford+$t\times T$ circuits, i.e., circuits with $t$ $T$ gates and any number of Clifford gates.
\begin{equation*}
    T=\begin{pmatrix}
        1&0\\
        0&\frac{1 + i}{\sqrt{2}}
    \end{pmatrix} \quad
    H=\frac{1}{\sqrt{2}}\begin{pmatrix}
        1&1\\
        1&-1
    \end{pmatrix} \quad 
    CZ=\left(\begin{smallmatrix}
    1&0&0&0\\
    0&1&0&0\\
    0&0&1&0\\
    0&0&0&-1
    \end{smallmatrix}\right)
    \quad
    CNOT=\left(\begin{smallmatrix}
    1&0&0&0\\
    0&1&0&0\\
    0&0&0&1\\
    0&0&1&0
    \end{smallmatrix}\right)
\end{equation*}
When applied to computational-basis state $\ket{x}$ with $x\in \{0, 1\}^n$, the $CZ_{a,b}$ ($CNOT_{a,b}$) gate applies an $Z_b$ gate ($X_b$ gate) to $\ket{x}$ if $x_a = 1$ and does nothing otherwise.
Finally, the SWAP gate $SWAP_{a,b} = CNOT_{a, b} CNOT_{b,a} CNOT_{a,b}$ swaps two qubits.

Extracting information from a register of $n$ quantum bits can be done by measurement. First-qubit computational-basis measurement is a random process that returns $j \in \{0, 1\}$ with probability $\sum_{l\in\{0,\dots,2^{n-1}-1\}}|v_{l+2^{n-1}j}|^2$. By (strong) simulation of a quantum circuit, we mean: given an $n$-qubit quantum circuit $U$ and some $j\in \{0, 1\}$, find the probability of obtaining outcome $j$ when measuring $U\ket{0}^{\otimes n}$~\cite{chen2023qseth}.

Instead of as a vector $\vec{v}$, a quantum state can also be represented as a density matrix $\ket{v} \cdot \bra{v}$ where $\bra{v}$ is the complex transpose of the vector $\ket{v}$.
For example, the state $\ket{v}=\frac{1}{\sqrt{2}}\left(\begin{smallmatrix}
    1\\i
\end{smallmatrix}\right)$ has density matrix $\frac{1}{\sqrt{2}}\left(\begin{smallmatrix}
    1\\i
\end{smallmatrix}\right)
\cdot
\frac{1}{\sqrt{2}}(1,-i) = \frac{1}{2}\left(\begin{smallmatrix}
    1&-i\\
    i&1
\end{smallmatrix}\right).$
Applying a quantum gate $U$ to a density matrix $\rho$ becomes conjugation: $\rho \mapsto U\cdot \rho \cdot U^{\dagger}$ where $U^\dagger$ is the conjugate transpose of $U$.

By the set $\mathcal{P}_n$ of $n$-qubit Pauli strings, we mean a length-$n$ Kronecker-product string of single-qubit Pauli gates, e.g. $X \otimes Y \otimes I$ is a $3$-qubit Pauli string.
The $n$-qubit Pauli group $\mathfrak{P}_n$ is the set $\{ \alpha P \mid \alpha \in \{\pm 1, \pm i\}, P \in \mathcal{P}_n\}$.
A (Pauli) stabilizer of a state $\ket{\phi}$ is an element $P \in \mathfrak{P}_n$ such that $P \ket{\phi} = \ket{\phi}$. We will write the set of Pauli stabilizers of $\ket{\phi}$ as $S(\ket{\phi})$ or shortly $S$ if $\ket{\phi}$ is clear. The set of Pauli stabilizers is always generated by $\log_2(|S|)$ generating Pauli string stabilizer~\cite{aaronson2004improved}.

\subsection{Decision diagrams for quantum computing}
\label{sec:qdd}

A decision diagram is a classical data structure for representing exponentially-sized vectors (i.e., pseudo-Boolean functions). Equivalent parts of a vector are stored only once, thereby compressing the representation of a vector exponentially in the best case. There exist many versions of decision diagrams, see for example for an overview~\cite{thanos2024automated}.
Quantum decision diagrams can represent quantum states by representing complex vectors.
In this paper, we will focus on three types, namely Multi Terminal Boolean Decision Diagrams (MTBDD)~\cite{fujita1997multi,bahar1993algebraic,viamontes2003improving}, Edge Valued Boolean Decision Diagrams (EVDD)~\cite{lai1994evbdd,tafertshofer1994factored,tafertshofer1997factored,wilson2005decision,sanner2005AffineADDs,miller2006qmdd,wang2008xqdd,fargier2013semiring,hong2022tensor}, and Local Invertible Map Decision Diagrams (LIMDD)~\cite{vinkhuijzen2023limdd,vinkhuijzen2023efficient,vinkhuijzen2024a}.
In general, we will refer to these three as \concept{multiplicative edge valued decision diagrams}.

Our explanation of decision diagrams is summarized in \Cref{fig:DDvariants}. On the left, we have a vector that we want to represent as a decision diagram. 

First, we transform this vector to a decision tree. A decision tree is a binary tree where each edge is a zero-edge (dotted) or a one-edge (straight). A path in the decision tree points to the value of the vector entry with binary index corresponding to the zero and one-edges of the path. Note that in \Cref{fig:DDvariants} the vector entries with value 0 are omitted for brevity.

The red and green node in the decision tree represent exactly the same $2$-vector $(1, 0)^{\top}$. Therefore, we can merge them as is shown in \Cref{fig:DDvariants}. This leads to an MTBDD, where all nodes representing the same vector are merged.

A EVDD is an edge-valued decision diagram. Here, numerical values are put on the edges instead of on the leaves. Therefore, nodes that are equal up to a multiplicative constant can be merged. In \Cref{fig:DDvariants} we see that the purple, red-green, and blue node are the same up to a factor $2$, $1$ and $-2$. Hence, these nodes are merged and the factors are set on the edges.

\begin{figure}
    \centering
    \tikzset{every picture/.style={->,thick}}

\begin{tikzpicture}[
    scale=0.3,
    every path/.style={>=latex},
    every node/.style={},
    inner sep=1pt,
    minimum size=0.3cm,
    line width=1pt,
    node distance=.8cm,
    thick,
    font=\scriptsize
    ]

\colorlet{origgreen}{green}
\colorlet{origblue}{blue}
\colorlet{origpurple}{purple}
\colorlet{origred}{red}

\colorlet{green}{ForestGreen}
\colorlet{blue}{RoyalBlue}
\colorlet{purple}{Orchid}
\colorlet{red}{Red}

    \node[draw,circle] (a1) {$x_1$};
    \node[draw,circle, below = .4cm of a1, xshift=-.65cm] (a2) {$x_2$};
    \node[draw,circle, below = .4cm of a1, xshift= .65cm] (a3) {$x_2$};
    
    \node[draw,circle, below = .4cm of a2, xshift=-.3cm,fill={red}] (a41) {$x_3$};
    \draw[e0=0] (a2) edge  node[] {} (a41);
    \node[draw,circle, below = .4cm of a2, xshift= .3cm,fill=purple] (a42) {$x_3$};
    \node[draw,circle, below = .4cm of a3, xshift=-.3cm,fill=green] (a43) {$x_3$};
    \node[draw,circle, below = .4cm of a3, xshift= .3cm,fill=blue] (a44) {$x_3$};

    \node[leaf, below=.35cm of a41, xshift=-0cm      ] (w1) {$1$};
    \node[leaf, below=.35cm of a42,inner sep=0pt] (w3) {$2$};
    \node[leaf, below=.35cm of a43,inner sep=0pt] (w5) {$1$};
    \node[leaf, below=.35cm of a44,inner sep=0pt] (w7) {$-2$};

    \draw[<-] (a1) --++(90:2cm) node[right,pos=.7] {$f$};
    \draw[e0 = 0] (a1) edge  node[] {} (a2);
    \draw[e1 = 0] (a1) edge  node[] {} (a3);

    \draw[e1=  0] (a2) edge  node[] {} (a42);
    \draw[e0=  0] (a3) edge  node[] {} (a43);
    \draw[e1=  0] (a3) edge  node[] {} (a44);

    \draw[e0=  0] (a41) edge  node[] {} (w1);
    \draw[e0=  0] (a42) edge  node[] {} (w3);
    \draw[e0=  0] (a43) edge  node[] {} (w5);
    \draw[e0=  0] (a44) edge  node[] {} (w7);

    \node[below= 2.8cm of a1]   (dt)  {Decision tree};

\node[left = 1.4cm of a1.south,yshift=-1cm] (vec) {
    \begin{minipage}{1.5cm}\scriptsize
    $\def\arraystretch{1.3}
    \begin{matrix}
	{000:}\\
        {001:}\\
        {010:}\\
        {011:}\\
        {100:}\\
        {101:}\\ 
        {110:}\\
        {111:}\\
    \end{matrix}
    \begin{bmatrix}
	{\color{red} 1}\\
        {\color{red} 0}\\
        {\color{purple} 2}\\
        {\color{purple} 0}\\
        {\color{green} 1}\\
        {\color{green} 0}\\ 
        {\color{blue}-2}\\
        {\color{blue} 0}\\
    \end{bmatrix}$
    \end{minipage}
};

    \node[draw,circle, right=2.5cm  of a1] (a1) {$x_1$};
    \node[draw,circle, below = .4cm of a1, xshift=-.45cm] (a2) {$x_2$};
    \node[draw,circle, below = .4cm of a1, xshift= .45cm] (a3) {$x_2$};
    
    \node[draw,circle, below = .4cm of a2, xshift= -.2cm,fill=purple] (a42) {$x_3$};
    \node[draw,circle, below = .4cm of a3, xshift=-.45cm, left color=red, right color=green] (a43) {$x_3$};
    \node[draw,circle, below = .4cm of a3, xshift= .2cm,fill=blue] (a44) {$x_3$};

    \node[leaf, below=.35cm of a42,inner sep=0pt] (w3) {$2$};
    \node[leaf, below=.35cm of a43,inner sep=0pt] (w5) {$1$};
    \node[leaf, below=.35cm of a44,inner sep=0pt] (w7) {$-2$};

    \draw[<-] (a1) --++(90:2cm) node[right,pos=.7] {$f$};
    \draw[e0 = 0] (a1) edge  node[] {} (a2);
    \draw[e1 = 0] (a1) edge  node[] {} (a3);

    \draw[e0=  0] (a2) edge  node[] {} (a43);
    \draw[e1=  0] (a2) edge  node[] {} (a42);
    \draw[e0=  0] (a3) edge  node[] {} (a43);
    \draw[e1=  0] (a3) edge  node[] {} (a44);
    \draw[e0=  0] (a42) edge  node[] {} (w3);
    \draw[e0=  0] (a43) edge  node[] {} (w5);
    \draw[e0=  0] (a44) edge  node[] {} (w7);

    \node[below= 2.8cm of a1]   (dt)  {MTBDD}; %

    \node[draw,circle, right=2.5cm  of a1] (a1) {$x_1$};
    \node[draw,circle, below = .4cm of a1, xshift=-.85cm] (a2) {$x_2$};
    \node[draw,circle, below = .4cm of a1, xshift= .85cm] (a3) {$x_2$};
    
    \node[draw,circle, below = .4cm of a2, xshift=.85cm,shade, shading=axis, left color=purple,  middle color=purple, right color=blue, shading angle=90] (a43) {$x_3$};

    \node[leaf, below=.35cm of a43      ] (w5) {$1$};

    \draw[<-] (a1) --++(90:2cm) node[right,pos=.7] {$f$};
    \draw[e0 = 0] (a1) edge  node[] {} (a2);
    \draw[e1 = 0] (a1) edge  node[] {} (a3);

    \draw[e0=15] (a2) edge  node[below left] {$1$} (a43);
    \draw[e1=15] (a2) edge  node[above right] {$2$} (a43);
    \draw[e0=20] (a3) edge  node[above left] {$1$} (a43);
    \draw[e1=20] (a3) edge  node[right] {$-2$} (a43);
    \draw[e0=  0] (a43) edge  node[] {} (w5);

    \node[below= 2.8cm of a1]   (dt)  {EVDD}; %

    \node[draw,circle, right=1.6cm  of a1] (a1) {$x_1$};
    \node[draw,circle, below = .4cm of a1 ] (a2) {$x_2$};
    \node[draw,circle, below = .4cm of a2 ] (a43) {$x_3$};
    \node[leaf, below=.35cm of a43      ] (w5) {$1$};

    \draw[<-] (a1) --++(90:2cm) node[right,pos=.7] {$f$};

    \draw[e0=20] (a1) edge  node[left] {} (a2);
    \draw[e1=20] (a1) edge  node[right,pos=.3] {~$Z\otimes I$} (a2);
    \draw[e0=20] (a2) edge  node[left] {} (a43);
    \draw[e1=20] (a2) edge  node[right,pos=.3] {~$2\cdot I$} (a43);
    \draw[e0=  0] (a43) edge  node[] {} (w5);

    \node[below= 2.8cm of a1]   (dt)  {LIMDD};

\end{tikzpicture}

\colorlet{green}{origgreen}
\colorlet{blue}{origblue}
\colorlet{purple}{origpurple}
\colorlet{red}{origred}
    \caption{Examples of decision diagram representations for a vector (left). For convenience, the edges pointing to 0 are omitted. Edges in EVDD and LIMDD without label have label 1 resp. $1\cdot I^{\otimes m}$. Figure reproduced from~\cite{thanos2024automated}.}
    \label{fig:DDvariants}
\end{figure}
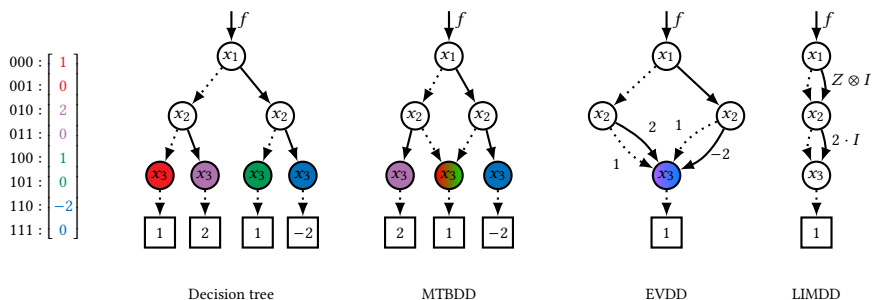

A LIMDD, an even more compressed decision diagram, is also an edge-valued decision diagram, but has `Pauli LIMs' on the edge: Pauli strings multiplied with a complex number. 
If two nodes are equal up to multiplication with a Pauli LIM, these nodes can be merged, and the corresponding Pauli LIM is put on the edge. In \Cref{fig:DDvariants} we have that the nodes labeled in EVDD with $x_2$ are equal up to application of a Pauli LIM, because $\left(\begin{smallmatrix}1&0&2&0\end{smallmatrix}\right)^T = Z\otimes I \cdot \left(\begin{smallmatrix}1&0&-2&0\end{smallmatrix}\right)^T$. Hence, these nodes are merged in the LIMDD representation and $Z\otimes I$ is set on the edge.

For all three decision diagram variants above, simulation algorithms have been proposed~\cite{viamontes2003improving,zulehner2018advanced,vinkhuijzen2023limdd}, i.e., algorithms that describe how to update a decision diagram representation of a quantum state vector after an arbitrary gate, or how to extract a random measurement outcome from the representation.

Finally, the three decision diagrams above are all canonical: two quantum-state vectors are represented by the same diagram (root) node if and only if the vectors are \emph{equivalent}.
Here, equivalence is defined differently for each diagram type, as should be evident from \autoref{fig:DDvariants}. Vectors are considered equivalent when:
\vspace{-.5em}
\begin{description}[noitemsep,itemsep=0pt,topsep=0pt,partopsep=0pt,parsep=0pt,]
\item[MTBDD:] equal.
\item[EVDD:] equal up to a constant factor $\gamma\in\complex$.
\item[LIMDD:] equal up to an operator $\gamma P$, where $\gamma\in\complex$ and Pauli string $P\in \mathcal{P}_n$.
\end{description}
\vspace{-.5em}
This explains the name multiplicative edge valued decision diagram.
Note finally that the equivalence modulus ($\gamma$ or $\gamma P$) is stored as edge label of the root node. 

To achieve canonicity, each node must satisfy several \emph{canonicity rules}, also called normalization rules.
We consider two types of canonicity rules for an $n$-qubit node of an EVDD and LIMDD:
\begin{itemize}
    \item `low': the dotted edge is labelled $1$ or $0$ for EVDD and $I^{\otimes n}$ for LIMDD;
    \item `L2': the dotted edge label $\alpha \in \mathbb{R}_{\geq 0}$ ( for LIMDD: $\alpha I^{\otimes n}$ and $\alpha \in \mathbb{R}_{>0}$) and solid edge label $\beta \in \mathbb{C}$ ($\beta P$ with $P \in \mathcal{P}_n$ for LIMDD) satisfy $|\alpha|^2 + |\beta|^2 = 1$ (that is, the vector $(\alpha, \beta)$ is normalized in the L2-norm).
\end{itemize}
We will denote the DD variants as EVDD-low, EVDD-L2, LIMDD-low and LIMDD-L2. The `L2' canonicity rules turn out to be more stable when edge values are represented by floating point numbers~\cite{brand2025q}, and are therefore often used in implementations of decision diagrams. We will use the `low' canonicity rule for our scaling guarantees for edge labels in \autoref{sec:coeff-bounds}, because they are easier to study analytically.

For a thorough understanding of the reasoning we often apply here, it is important that low canonicity is obtained by dividing out the edge value on the low edge as the following figure illustrates. This ensures that the left edge has always value 1, except in the case when $A=0$ which is already canonical. For LIMDD, the canonical left edge always has value $I^{\otimes n}$, since the left edge and right edge can be swapped by an $X$ LIM on the root edge.

\begin{center}
\tikzset{every picture/.style={->,thick}}

\begin{tikzpicture}[
    scale=0.3,
    every path/.style={>=latex},
    every node/.style={},
    inner sep=1pt,
    minimum size=0.3cm,
    line width=1pt,
    node distance=.8cm,
    thick,
    font=\scriptsize
    ]

    \node[draw,circle] (a1) {$x_1$};
    \node[draw,circle, below = .4cm of a1, xshift=-.65cm] (a2) {$x_2$};
    \node[draw,circle, below = .4cm of a1, xshift= .65cm] (a3) {$x_2$};
    
    \draw[<-] (a1) --++(90:2cm) node[right,pos=.7] {$C$};
    \draw[e0 = 0] (a1) edge  node[above left] {$A$} (a2);
    \draw[e1 = 0] (a1) edge  node[above right] {$B$} (a3);

    \node[draw,circle,right= 5cm of a1] (a1a) {$x_1$};
    \node[draw,circle, below = .4cm of a1a, xshift=-.65cm] (a2) {$x_2$};
    \node[draw,circle, below = .4cm of a1a, xshift= .65cm] (a3) {$x_2$};
    
    \draw[<-] (a1a) --++(90:2cm) node[right,pos=.7] {$CA$};
    \draw[e0 = 0] (a1a) edge  node[above left] {$1\ (I^{\otimes n})$} (a2);
    \draw[e1 = 0] (a1a) edge  node[above right] {$A^{-1}B$} (a3);

    \draw[->,dashed, shorten >=1cm,  shorten <=1cm] (a1) edge node[above] {low canonicity rule}  (a1a);
    
\end{tikzpicture}
\end{center}

\section{Problem statement and related work\label{sec:motivation}}
We start by providing an example how floating point inaccuracy can lead to a blow-up of decision diagrams.
In \autoref{fig:motivating_example}, a quantum circuit is simulated using an EVDD-low state representation.
Using an exact method, which represents complex numbers on the edges symbolically, the output state's EVDD has only three nodes.
In the final step, the EVDD canonization procedure computes an edge value $a := \frac{1-\omega^2\cdot(-i)}{1+\omega^2\cdot(-i)}$, which should equal $0$ since $\omega^2 = i$.
In contrast, in the figure, we see that when insufficient precision is used, $a$ cannot be computed exactly. 
As a consequence, there now exist two EVDD nodes at the one-qubit level which are not EVDD-equivalent and hence cannot be merged given the used precision, yielding a diagram of four nodes.
Our experiments in \autoref{sec:experiments} show that the size blow-up due to limited precision also occurs in practice when simulating real-world quantum circuits.
In fact, the size of the blow-up is a limiting factor for simulating large quantum circuits.

\begin{figure}[t]
    \centering
    \tikzset{every picture/.style={->,thick}}

\begin{tikzpicture}[
    scale=0.3,
    every path/.style={>=latex},
    every node/.style={},
    inner sep=1pt,
    minimum size=0.3cm,
    line width=1pt,
    node distance=.8cm,
    thick,
    font=\scriptsize
    ]

    \node[draw,circle, right=1.6cm  of a1] (a1) {$x_1$};
    \node[draw,circle, below = .4cm of a1 ] (a2) {$x_2$};
    \node[leaf, below=.35cm of a2      ] (w5) {$1$};

    \node[below = .2 of w5] (startedge) {};
    \node[right = .8cm of startedge] (finaledge) {};
    \draw[->,dashed] (startedge) -- (finaledge) node[below,pos=.5]{$H_1$};
    
    \draw[<-] (a1) --++(90:2cm) node[right,pos=.7] {1};

    \draw[e0=20] (a1) edge  node[left,pos=.3] {1} (a2);
    \draw[e1=20] (a1) edge  node[right,pos=.3] {0} (a2);
    \draw[e0=20] (a2) edge  node[left,pos=.3] {1} (w5);
    \draw[e1=20] (a2) edge  node[right,pos=.3] {0} (w5);

    \node[draw,circle, right=.6cm  of a1] (a1) {$x_1$};
    \node[draw,circle, below = .4cm of a1 ] (a2) {$x_2$};
    \node[leaf, below=.35cm of a2      ] (w5) {$1$};

    \draw[<-] (a1) --++(90:2cm) node[right,pos=.7] {$\frac{1}{\sqrt{2}}$};

    \draw[e0=20] (a1) edge  node[left,pos=.3] {1} (a2);
    \draw[e1=20] (a1) edge  node[right,pos=.3] {1} (a2);
    \draw[e0=20] (a2) edge  node[left,pos=.3] {1} (w5);
    \draw[e1=20] (a2) edge  node[right,pos=.3] {0} (w5);
    
    \node[below = .2 of w5] (startedge) {};
    \node[right = .8cm of startedge] (finaledge) {};
    \draw[->,dashed] (startedge) -- (finaledge) node[below,pos=.5]{$T_1$};

    \node[draw,circle, right=.6cm  of a1] (a1) {$x_1$};
    \node[draw,circle, below = .4cm of a1 ] (a2) {$x_2$};
    \node[leaf, below=.35cm of a2      ] (w5) {$1$};

    \draw[<-] (a1) --++(90:2cm) node[right,pos=.7] {$\frac{1}{\sqrt{2}}$};

    \draw[e0=20] (a1) edge  node[left,pos=.3] {1} (a2);
    \draw[e1=20] (a1) edge  node[right,pos=.3] {$\omega$} (a2);
    \draw[e0=20] (a2) edge  node[left,pos=.3] {1} (w5);
    \draw[e1=20] (a2) edge  node[right,pos=.3] {0} (w5);

    \node[below = .2 of w5] (startedge) {};
    \node[right = .8cm of startedge] (finaledge) {};
    \draw[->,dashed] (startedge) -- (finaledge) node[below,pos=.5]{$T_1$};

    \node[draw,circle, right=.6cm  of a1] (a1) {$x_1$};
    \node[draw,circle, below = .4cm of a1 ] (a2) {$x_2$};
    \node[leaf, below=.35cm of a2      ] (w5) {$1$};

    \draw[<-] (a1) --++(90:2cm) node[right,pos=.7] {$\frac{1}{\sqrt{2}}$};

    \draw[e0=20] (a1) edge  node[left,pos=.3] {1} (a2);
    \draw[e1=20] (a1) edge  node[right,pos=.3] {$\omega^2$} (a2);
    \draw[e0=20] (a2) edge  node[left,pos=.3] {1} (w5);
    \draw[e1=20] (a2) edge  node[right,pos=.3] {0} (w5);

    \node[below = .2 of w5] (startedge) {};
    \node[right = 1.cm of startedge] (finaledge) {};
    \draw[->,dashed] (startedge) -- (finaledge) node[below,pos=.5]{$S^\dagger_1$};

    \node[draw,circle, right=.8cm  of a1] (a1) {$x_1$};
    \node[draw,circle, below = .4cm of a1 ] (a2) {$x_2$};
    \node[leaf, below=.35cm of a2      ] (w5) {$1$};

    \draw[<-] (a1) --++(90:2cm) node[right,pos=.7] {$\frac{1}{\sqrt{2}}$};

    \draw[e0=20] (a1) edge  node[left,pos=.3] {1} (a2);
    \draw[e1=20] (a1) edge  node[right,pos=.3] {\tiny$\omega^2 \cdot(-i)$} (a2);
    \draw[e0=20] (a2) edge  node[left,pos=.3] {1} (w5);
    \draw[e1=20] (a2) edge  node[right,pos=.3] {0} (w5);

    \node[below = .2 of w5] (startedge) {};
    \node[right = 1.5cm of startedge] (finaledge) {};
    \draw[->,dashed] (startedge) -- (finaledge) node[below,pos=.5]{$H_1$};

    \node[draw,circle, right=1.3cm  of a1] (a1) {$x_1$};
    \node[draw,circle, below = .4cm of a1 ] (a2) {$x_2$};
    \node[leaf, below=.35cm of a2      ] (w5) {$1$};

    \draw[<-] (a1) --++(90:2cm) node[above] {$\frac{(1+\omega^2\cdot(-i))}{\sqrt{2}\cdot\sqrt{2}}$};

    \draw[e0=20] (a1) edge  node[left,pos=.3] {1} (a2);
    \draw[e1=20] (a1) edge  node[right,pos=.3] {$\frac{1-\omega^2\cdot(-i)}{1+\omega^2\cdot(-i)}=:a$} (a2);
    \draw[e0=20] (a2) edge  node[left,pos=.3] {1} (w5);
    \draw[e1=20] (a2) edge  node[right,pos=.3] {0} (w5);

    \node[ right=2.8cm  of a1] (tmp) {};
    
    \node[draw,circle, above=1.5cm  of tmp] (a1) {$x_1$};
    \node[draw,circle, below = .4cm of a1 ] (a2) {$x_2$};
    \node[leaf, below=.35cm of a2      ] (w5) {$1$};

    \node[left = 2.5cm of tmp] (connector1left) {};
    \node[left = .1cm of a2] (connector1right) {};
    \draw[->,dashed] (connector1left) -- (connector1right) node[text width = 1.5 cm,align=center,above left,pos=.9] {$CNOT_{1,2}$\\$a=0$};

    \draw[<-] (a1) --++(90:2cm) node[right] {$\frac{(1+\omega^2\cdot(-i))}{\sqrt{2}\cdot\sqrt{2}}$};

    \draw[e0=20] (a1) edge  node[left,pos=.3] {1} (a2);
    \draw[e1=20] (a1) edge  node[right,pos=.3] {0} (a2);
    \draw[e0=20] (a2) edge  node[left,pos=.3] {1} (w5);
    \draw[e1=20] (a2) edge  node[right,pos=.3] {0} (w5);

    \node[draw,circle, below=1.5cm  of tmp] (a1) {$x_1$};
    \node[ below = .4cm of a1 ] (a2) {};
    \node[draw,circle, left=.25 of a2] (a2left) {$x_2$};
    \node[draw,circle, right=.25 of a2] (a2right) {$x_2$};
    \node[leaf, below=.55cm of a2      ] (w5) {$1$};

    \node[below = 1. of connector1left] (connector2left) {};
    \node[left = .1cm of a1] (connector2right) {};
    \draw[->,dashed] (connector2left) -- (connector2right) node[text width = 1.5 cm,above right,align=center,pos=.3] {$CNOT_{1,2}$\\$a\not=0$};

    \draw[<-] (a1) --++(90:2cm) node[right] {$\frac{(1+\omega^2\cdot(-i))}{\sqrt{2}\cdot\sqrt{2}}$};

    \draw[e0=20] (a1) edge  node[left,pos=.3] {1} (a2left);
    \draw[e1=20] (a1) edge  node[above right,pos=.6] {$\frac{1-\omega^2\cdot(-i)}{1+\omega^2\cdot(-i)}$} (a2right);
    \draw[e0=20] (a2left) edge  node[left,pos=.3] {1} (w5);
    \draw[e1=20] (a2left) edge  node[right,pos=.3] {0} (w5);
    \draw[e0=20] (a2right) edge  node[left,pos=.3] {0} (w5);
    \draw[e1=20] (a2right) edge  node[right,pos=.3] {1} (w5);

\end{tikzpicture}
    \caption{
    Motivating example: numerical errors can lead to a larger decision diagram. The figure shows the intermediate EVDDs of the $2$-qubit circuit $H_1 T_1 T_1 S^\dagger_1 H_1 CNOT_{1,2}$ when applied to the input state $\ket{00}$, both with an exact representation of the complex numbers on the edges (above) and with a floating point arithmetic error so that the value
    $a := \frac{1-\omega^2\cdot(-i)}{1+\omega^2\cdot(-i)}$ is not exactly equal to $0$ (below).
    Here, $\omega = \frac{1}{\sqrt{2}}(1+i)\approx 0.707+0.707i$.
    }
    \label{fig:motivating_example}
\end{figure}
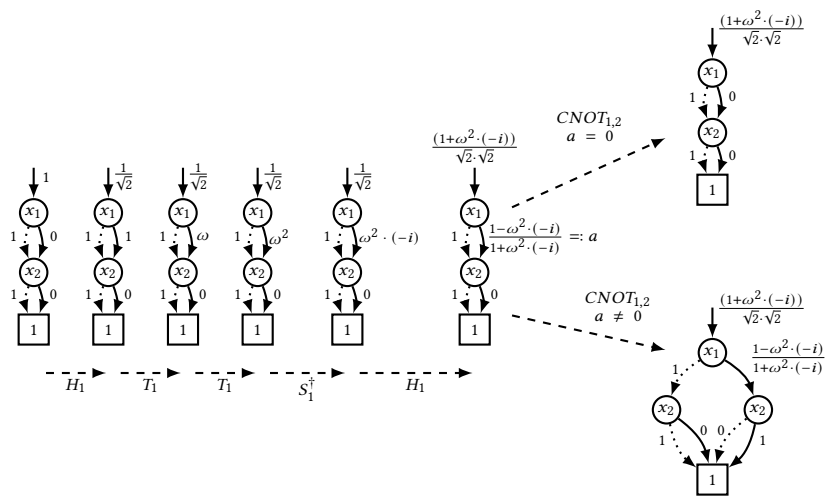

\subsection{Related work: numerical accuracy}

The existing approaches to tackling the accuracy issue can roughly be divided along two lines.

First, to introduce a \textbf{tolerance} $\epsilon$ and consider two decision diagram nodes equivalent if they are no more than $\epsilon$ apart.
The distance measure depends on the context.
This line of approaches leads to fewer nodes at the cost of introducing an error~\cite{zulehner2019efficiently,niemann2020overcoming}.
The approach has also been studied empirically from the perspective of making the decision diagram approximate~\cite{hillmich2022approximating}. Or, with stochastic reasoning instead of quantum computing as application in mind, approximation algorithms for MTBDDs~\cite{sanner2010approximate,st2000apricodd}, which are exponentially less succinct than EVDDs and LIMDDs, have been found, as well as for the powerful Affine Algebraic Decision Diagrams~\cite{sanner2010approximate}.
What is lacking in all these works is a \emph{guarantee} on the scaling of performance, e.g. bounds on the number of nodes as function of the allowed tolerance $\epsilon$.\footnote{Exception is the error analysis of~\cite{brand2025numerical}. However, they only provide a guarantee for applying only one gate with MTBDD.} Though there exist some guarantees on the total numerical error of DD related tasks such as model counting~\cite{bryant2025numerical}, this is still not exact and may therefore still lead to a blow-up of the decision diagram.
We expect that the error-propagation analysis for relevant distance metrics (such as fidelity between the desired and real outputted quantum-state vector) will be generally hard for EVDD, LIMDD (and AADD) and might only be practically feasible on a case-by-case (quantum-circuit dependent) basis. 

The other approach is to make use of an \textbf{algebraic representation} of the complex numbers in the decision diagram.
This algebraic approach is also used in model counting~\cite{Eiter_Kiesel_2021,kimmig2017algebraic}, with applications in quantum computing~\cite{DECAMPOS2020104627,Mei_2024}.
Existing work for decision diagrams has focused on the universal~\cite{dawson2005solovaykitaev} Clifford+$T$ gate set: if a quantum circuit with only Clifford+$T$ gates is applied to the all-zero state $\ket{0}^{\otimes n}$ for some $n \in \{1, 2, \dots \}$ as input, then the entries of the quantum state vector at any point during the computation can be written as elements of the ring $\mathbb{Z}[i, \frac{1}{\sqrt{2}}] = \{\frac{1}{2^n} \left(a + bi + c \sqrt{2} + di \sqrt{2} \right) \mid a, b, c, d \in \mathbb{Z}, n \in \mathbb{N}\}$~\cite{giles2013exact}.
Parameterizing this ring differently, as
\begin{equation*}
    \label{eq:dyadic-fractions}
	\left\{\frac{1}{\sqrt{2^k}} \left(a + b\omega + c \omega^2 + d \omega^3\right) \mid a, b, c, d \in \mathbb{Z}, k \in \mathbb{N}\right\}\text{ where $\omega = \frac{1 + i}{\sqrt{2}}$}
,
\end{equation*}
an EVDD implementation was able to empirically characterize the cost or gain in terms of number of nodes for various tolerance values~\cite{zulehner2019efficiently,niemann2020overcoming}.
Indeed, for some scenarios, the algebraic approach was slower than the $\epsilon$-tolerance approach.
The same parameterization has also been used in other works on quantum-circuit simulation using Binary Decision Diagrams~\cite{tsai2021bit,wei2022accurate}
and others~\cite{chen2023theory,chen2023autoq,chen2023automata} because a discrete complex-number parameterization was needed in these works.
However, all these works are again empirical, lacking any scaling guarantees.

\subsection{Related work: gate-count dependent decision diagram node number bounds}
\label{sec:related_word--DD-size}

Existing size bounds on the number of nodes in decision diagrams for quantum computing tackle fixed families of states, see e.g. ~\cite{vinkhuijzen2024a} and references therein, and only give bound on the maximal increase of nodes when a single gate is applied~\cite{vinkhuijzen2024a}.
To the best of our knowledge, we are the first to provide parameterized complexity bounds on the number of nodes of a LIMDD and EVDD, showing that simulation is fixed-parameter tractable in the types of gates occuring in the quantum circuit.
The LIMDD upper bound $O(2^t)$ in particular is fixed-parameter tractable in $t$, the $T$-count.

For comparison, the extended stabilizer formalism (also called stabilizer rank simulation method), which is a current best known simulation method for Clifford+$T$ simulation, simulates Clifford+$T$ circuits in $\Theta(2^{\alpha\cdot t})$ with $\alpha\leq0.3963$~\cite{bravyi2019simulation,qassim2021improved}, which is an exponential lower bound.
However, our LIMDD size bound is an upper bound that is not tight:  it is known that the Clifford+$T$ circuit which outputs an $n$-qubit W-state has $\Omega(n)$ $T$ gates, whereas the LIMDD has width $n$ rather than $2^n$.
This example shows an exponential separation with the extended stabilizer formalism for Clifford+$T$ in favor of LIMDD~\cite{vinkhuijzen2023limdd}.

\section{Bounds on the size of coefficients in decision diagrams}
\label{sec:coeff-bounds}

In this section, we consider the edge coefficients of `low'-canonical decision diagrams for handling Clifford+$T$ circuits (see also the preliminaries on decision diagrams in \Cref{sec:qdd}).
We first focus on the coefficients in the decision diagram of the output state of such a circuit.
Specifically, in \Cref{sec:accuracy-bound-static}, we prove that the number of bits needed to express these coefficients scales only in the number of $T$ gates and the number of qubits, but is independent of the number of Clifford gates.
Next, in \Cref{subsec:measurement_probabilities_are_in_qisq}, we extend the proof to coefficients encountered in applying gates and measurements while simulating the circuit.
The section contains the main reasoning and proof sketches, referring to proof details in \Cref{appendix:coeffs-bound}.

\subsection{Decision-diagram representation needs only succinct numbers \label{sec:accuracy-bound-static}}
Here, we prove \Cref{thm:edgelabelbound}, which bounds the complexity of the decision-diagram coefficients for output states of Clifford+$t\times T$ circuits.
We denote the states of interest, i.e., the set of all $n$-qubit quantum states $\ket{\phi} = U \ket{0}^{\otimes n}$ where $U$ is a Clifford+$t\times T$ unitary, by $\cktstates^{n}_t$.
The proof is divided into two steps.
First, \Cref{lemma:vector-entry-lemma} proves that the density matrix entries of a quantum state in $\cktstates^n_t$ take a simple form.
Then, \Cref{thm:edgelabelbound} follows by the observation that edge coefficients of a decision diagram are ratios of density matrix entries.

\begin{lemma}[Density matrix entry lemma]%
\label{lemma:vector-entry-lemma}
For integer $n>0, t\geq0$ we define the set of complex numbers $Q_{n, t}$ as
\begin{multline*}
Q_{n,t}=\biggl\{\frac{\ell}{\sqrt{2^t}}+\frac{m}{\sqrt{2^{t-1}}} + i \left(\frac{\ell'}{\sqrt{2^t}}+\frac{m'}{\sqrt{2^{t-1}}}\right)
\biggm|
\ell, \ell'\in[-2^{n+t},2^{n+t}] 
\textnormal{ and } m, m'\in[-2^{n+t-1},2^{n+t-1}]  \\ \textnormal{ and }(t=0\implies m=m'=0)\biggr\}
\end{multline*}
where $[a,b]$ denotes set $\{a,a+1,\dots,b-1,b\}$ for any $a,b\in\mathbb{Z}$.\\
If $\ket{\phi} \in \cktstates^n_t$ and $\phi=\dyad{\phi}$ is the density matrix of $\ket{\phi}$, then $2^n \cdot \langle y\mid\phi\mid x\rangle \in Q_{n,t}$ for all $x,y\in\{0,\dots,2^n-1\}$. That is, every density matrix entry multiplied by $2^n$ is a member of the set $Q_{n,t}$.
\end{lemma}
\begin{proof}[Proof of \Cref{lemma:vector-entry-lemma}]

Our initial step is to characterize the Pauli basis coefficients $c_P$ of $\phi$, i.e., we decompose $\phi$ in the Pauli basis as $\phi = \dyad{\phi} =\frac{1}{2^n}\sum_{P\in \mathcal{P}_n}c_P P$, where $c_P=\text{Tr}(P^{\dagger}\cdot\phi)$~\cite{jones2024decomposing}.
	To characterize the values $c_P$ given that $\phi$ was produced by a Clifford$+t\times T$ circuit (i.e., $\phi \in \cktstates_n^t$), we first consider the circuit's input state $\ket{0}^{\otimes n}$: as this is a stabilizer state, we know that $c_P\in\{0,1,-1\}$ for all $P$~\cite{aaronson2004improved}.
	Next, consider the action $\rho \mapsto U\rho U^{\dagger}$ of each of the individual gates $U$ in the circuit that produced $\phi$.
To start, Clifford gates permute the Pauli strings but otherwise do not change the individual $c_P$ except by potentially a factor $\pm 1$.
	Next, the $T$ gate has the property that $TI T^{\dagger} = I$ and $TZT^{\dagger} = Z$, i.e., leaves $I$ and $Z$ invariant under conjugation.
	In contrast, the $T$ gate maps $X$ to $\frac{X+Y}{\sqrt{2}}$ and $Y$ to $\frac{Y-X}{\sqrt{2}}$. Thus, applying a $T$-gate leaves Pauli-basis coefficients of the form $c_{I\otimes P'}$ and $c_{Z\otimes P'}$ invariant, but it changes coefficients $c_{X\otimes P'}$ and $c_{Y\otimes P'}$, where $P'$ is a Pauli string of length $n-1$.
Now, treating each of the terms $c_P$ separately, a straightforward induction shows that every $c_P$ can be written as $c_P= 0 \text{ or }\pm1$ for $t=0$, and $c_P=\frac{\ell}{\sqrt{2^t}}+\frac{m}{\sqrt{2^{t-1}}}$ with $\ell\in[-2^t,2^t], m\in[-2^{t-1},2^{t-1}]$ for $t>0$, see \Cref{lem:C_k-representation} for details.

	Next, we note that for all $x, y \in [0, 2^n - 1]$, the entry of $\phi$ at row $x$ and column $y$ can be written as
\begin{eqnarray*}
2^n \cdot \langle y\mid\phi\mid x\rangle 
= \sum_{P\in \mathcal{P}_n}c_P \langle y\mid P\mid x\rangle
=\underbrace{\sum_{P\in \mathcal{P}_n: \langle y \mid P \mid x\rangle \in \{0, \pm 1\}}c_P \langle y\mid P\mid x\rangle}_{\textnormal{sum of $\leq 2^n$ values $c_P \cdot \pm 1$ }}
+
\underbrace{\sum_{P\in \mathcal{P}_n: \langle y \mid P\mid x\rangle \in \{0, \pm i\}}c_P \langle y\mid P\mid x\rangle}_{\textnormal{sum of $\leq 2^n$ values $c_P \cdot \pm i$ }}
\\
\end{eqnarray*}
which is of the form
$\frac{\ell}{\sqrt{2^t}}+\frac{m}{\sqrt{2^{t-1}}} + i(\frac{\ell'}{\sqrt{2^t}}+\frac{m'}{\sqrt{2^{t-1}}})$
 with $\ell,\ell'\in[-2^{n+t},2^{n+t}], m,m'\in[-2^{n+t-1},2^{n+t-1}]$
and hence $2^n\cdot\langle y \mid \phi\mid x \rangle \in Q_{n,t}$.
Here, we used the fact that of the $4^n$ elements $P$ of $\mathcal{P}_n$, only $2^n$ satisfy $\langle y\mid P\mid x\rangle \neq 0$; these $2^n$ satisfy $\langle y\mid P\mid x\rangle \in\{\pm1,\pm i\}$, bounding the number of such $P$ in $\mathcal{P}_n$ to at most $2^n$ for both the $\langle y\mid P\mid x\rangle \in\{\pm1\}$ as the $\langle y\mid P\mid x\rangle \in\{\pm i\}$ case. 
\end{proof}

On our way to proving \Cref{lemma:vector-entry-lemma}, we define the following subsets of the ring $\mathbb{Q}[i, \sqrt{2}]$.
\footnote{As $k$ approaches infinity, the set $R_k$ approaches the ring $\mathbb{Q}[i,\sqrt{2}]$: we have $\mathbb{Q}=\bigcup_{k\geq0}\qisq_k$, and hence $\mathbb{Q}[i,\sqrt{2}]=\bigcup_{k\geq0}R_k$.
Note, however, that for each $k$, $R_k$ is a set, but not a ring, as it is not closed under addition.}

\begin{definition}%
\label{def:qisqnumbers}
We define the finite sets $\left(R_k\right)_{k \in \mathbb{N}_+}$ of complex numbers with $R_k \subseteq R_{k+1}$, as
\begin{eqnarray*}
    R_k :=& \{ a + b \sqrt{2} + c i + d i \sqrt{2} \mid a, b, c, d \in \qisq_k\}\\
    \text{where}\quad
	\qisq_k:=& \left\{\frac{p}{q} \mid p,q\in [-2^k, 2^k],q\not=0\right\}\subset\mathbb{Q}%
	\label{eq:qisq}
\end{eqnarray*}
\end{definition}

We note that, as $p$ and $q$ can each be represented in $k + 2$ bits, each element of $R_k$ can be represented in $4 \cdot 2(k+2) = 8k + 16 = O(k)$ bits.
Also note that every $\frac{p}{q}\in\qisq_k$ can be uniquely (i.e., canonically) represented by applying Euclid's algorithm.

We are now ready to state the theorem, where we denote the unique edge to the root node as the \textit{root edge}, and refer to all other edges as \textit{bulk edges}.

\begin{theorem}[DD-representation needs only short numbers.]%
\label{thm:edgelabelbound}
Let $U$ be an $n$-qubit Clifford+$t\times T$ unitary.
The bulk edge values, as well as the magnitude of the root edge label, of both the EVDD-low or LIMDD-low representation of $\ket{\phi}=U\ket{0}^{\otimes n}$ are in the set $R_{2n+2t+1}$ and can thus be described in $O(n+t)$ bits.
\end{theorem}

\begin{proof}[Proof of \Cref{thm:edgelabelbound}]
First note that every edge label in the bulk of a low-EVDD or low-LIMDD is a quotient of two state vector entries $\langle x\mid\phi\rangle$ and $\langle y\mid\phi\rangle$ (for LIMDD: up to multiplication with $-1$, $i$ or $-i$ due to the Pauli LIMs), where the denominator of this quotient is always nonzero since the edge labels are finite complex numbers.
This is due to the fact that every edge label $a$ distinguishes between two paths with only `$1$' as edge labels (proof idea in \autoref{fig:quotient-of-statevectorentries}, full proof in \Cref{lem:DDlabel=fraction-of-C_k}). 

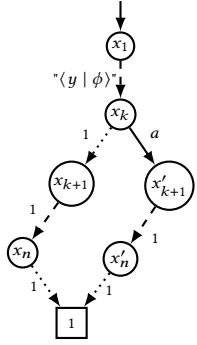
\begin{figure}[h]
    \centering
    \vspace{-1em}
    
        \tikzset{every picture/.style={->,thick}}
        \begin{tikzpicture}[
            scale=0.3,
            every path/.style={>=latex},
            every node/.style={},
            inner sep=1pt,
            minimum size=0.3cm,
            line width=1pt,
            node distance=.8cm,
            thick,
            font=\scriptsize
            ]
        
            \node[draw,circle] (root) {$x_1$};
            \node[draw,circle, below = .5cm of root] (a1) {$x_k$};
            \node[draw,circle, below = .4cm of a1, xshift=-.65cm] (a2) {$x_{k+1}$};
            \node[draw,circle, below = .4cm of a2, xshift=-.65cm] (a22) {$x_{n}$};
        
            \node[draw,circle, below = .4cm of a1, xshift= .65cm] (a3) {$x'_{k+1}$};
            \node[draw,circle, below = .4cm of a3, xshift=-.65cm] (a33) {$x'_{n}$};
            \node[draw,leaf, below = .4cm of a33, xshift=-.65cm] (a333) {1};
            
            \draw[<-] (root) --++(90:2cm) node[right,pos=.7] {};
            \draw[e0 = 0,->,dashed] (root) edge  node[left] {"$\langle y\mid\phi\rangle$"} (a1);
            \draw[e0 = 0] (a1) edge  node[above left] {1} (a2);
            \draw[e0 = 0,->,dashed] (a2) edge  node[above left] {1} (a22);
            \draw[e0 = 0] (a22) edge  node[left] {1} (a333);
            \draw[e1 = 0] (a1) edge  node[above right] {$a$} (a3);
            \draw[e0 = 0,->,dashed] (a3) edge  node[below right] {1} (a33);
            \draw[e0 = 0] (a33) edge  node[right] {1} (a333);
        
        \end{tikzpicture}
    
    \caption{Visualization of the proof that each EVDD-edge label is a quotient of state-vector entries. Let $a$ be an bulk edge label in the DD, and let $x_k$ be the parent node of this edge and $x'_{k+1}$ its child node. By the "low"-canonicity rule, every node has at least one outgoing edge with edge label `1'. Therefore, there exist paths from both $x_k$ to the terminal node and from $x'_{k+1}$ to the terminal node with only `1's as edge labels. Concatenating those paths with a path from the root node (i.e., $x_1$) to $x_k$, we have two paths from root to terminal representing the values $\langle x\mid\phi\rangle$ and $\langle y\mid\phi\rangle$ for some $x,y\in\{0,1\}^n$. As these paths differ only by the edge label $a$, we have $a=\frac{\langle x\mid\phi\rangle}{\langle y\mid\phi\rangle}$.}
    \label{fig:quotient-of-statevectorentries}
    \vspace{-1em}
\end{figure}

As $\frac{\langle x\mid\phi\rangle}{\langle y\mid\phi\rangle}=\frac{\langle x\mid\phi\rangle\langle \phi\mid y\rangle}{\langle y\mid\phi\rangle\langle \phi\mid y\rangle} = \frac{2^n\cdot\langle x\mid \rho\mid y\rangle}{2^n\cdot\langle y\mid \rho\mid y\rangle}$, we conclude by \Cref{lemma:vector-entry-lemma} that every edge label in the bulk of an EVDD or LIMDD is a quotient of two numbers in $Q_{n,t}$, with the denominator being real. In \Cref{lemma:ratio_of_state-vector_entries_is_linearly_representable}, we show by direct calculation that the quotient of two such elements in $Q_{n,t}$ is in $R_{2(n+t)+1}$. For LIMDD, the multiplication with $-1$, $i$ or $-i$ due to the Pauli LIMs gives the same result, because $a\in R_{m}$ implies $-a,ia,-ia\in R_{m}$ for all $a$ and for all $m\in\mathbb{N}$. This shows the theorem for the bulk edges.

The root edge value equals $\langle x\mid \phi\rangle$ for some $x\in\{0,1\}^n$, as there is a path with only values `1' in the bulk from the root node to the terminal node. The magnitude square of the root edge label can be efficiently stored because $|\langle x\mid \phi\rangle|^2=\langle x\mid\rho\mid x\rangle\in Q_{2n,t}\subseteq R_{2(n+t)+1}$.%
\end{proof}

\begin{example}
\autoref{fig:leadingExample} shows the EVDD resp. LIMDD representation of states that are reached during the simulation of the Clifford$+T$ circuit $H_1H_2CZ_{1,2}T_1H_1,T_2$. We see that the edge values are indeed in $R_{2n+2t+1}$, supporting the statement of \autoref{thm:edgelabelbound}.
\end{example}

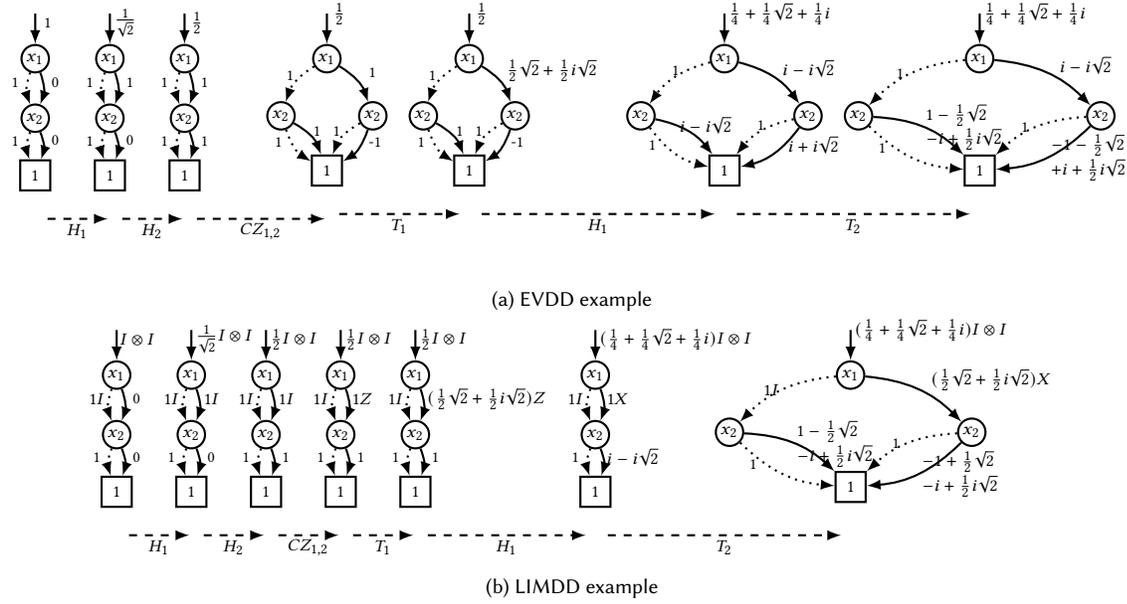
\begin{figure}
\centering
\begin{subfigure}{\textwidth}
    \centering
    \tikzset{every picture/.style={->,thick}}

\begin{tikzpicture}[
    scale=0.3,
    every path/.style={>=latex},
    every node/.style={},
    inner sep=1pt,
    minimum size=0.3cm,
    line width=1pt,
    node distance=.8cm,
    thick,
    font=\scriptsize
    ]

    \node[draw,circle, right=1.6cm  of a1] (a1) {$x_1$};
    \node[draw,circle, below = .4cm of a1 ] (a2) {$x_2$};
    \node[leaf, below=.35cm of a2      ] (w5) {$1$};

    \node[below = .2 of w5] (startedge) {};
    \node[right = .8cm of startedge] (finaledge) {};
    \draw[->,dashed] (startedge) -- (finaledge) node[below,pos=.5]{$H_1$};
    
    \draw[<-] (a1) --++(90:2cm) node[right,pos=.7] {1};

    \draw[e0=20] (a1) edge  node[left,pos=.3] {1} (a2);
    \draw[e1=20] (a1) edge  node[right,pos=.3] {0} (a2);
    \draw[e0=20] (a2) edge  node[left,pos=.3] {1} (w5);
    \draw[e1=20] (a2) edge  node[right,pos=.3] {0} (w5);

    \node[draw,circle, right=.6cm  of a1] (a1) {$x_1$};
    \node[draw,circle, below = .4cm of a1 ] (a2) {$x_2$};
    \node[leaf, below=.35cm of a2      ] (w5) {$1$};

    \draw[<-] (a1) --++(90:2cm) node[right,pos=.7] {$\frac{1}{\sqrt{2}}$};

    \draw[e0=20] (a1) edge  node[left,pos=.3] {1} (a2);
    \draw[e1=20] (a1) edge  node[right,pos=.3] {1} (a2);
    \draw[e0=20] (a2) edge  node[left,pos=.3] {1} (w5);
    \draw[e1=20] (a2) edge  node[right,pos=.3] {0} (w5);
    
    \node[below = .2 of w5] (startedge) {};
    \node[right = .8cm of startedge] (finaledge) {};
    \draw[->,dashed] (startedge) -- (finaledge) node[below,pos=.5]{$H_2$};

    \node[draw,circle, right=.6cm  of a1] (a1) {$x_1$};
    \node[draw,circle, below = .4cm of a1 ] (a2) {$x_2$};
    \node[leaf, below=.35cm of a2      ] (w5) {$1$};

    \draw[<-] (a1) --++(90:2cm) node[right,pos=.7] {$\frac{1}{2}$};

    \draw[e0=20] (a1) edge  node[left,pos=.3] {1} (a2);
    \draw[e1=20] (a1) edge  node[right,pos=.3] {1} (a2);
    \draw[e0=20] (a2) edge  node[left,pos=.3] {1} (w5);
    \draw[e1=20] (a2) edge  node[right,pos=.3] {1} (w5);

    \node[below = .2 of w5] (startedge) {};
    \node[right = 1.7cm of startedge] (finaledge) {};
    \draw[->,dashed] (startedge) -- (finaledge) node[below,pos=.5]{$CZ_{1,2}$};

    \node[draw,circle, right=1.5cm  of a1] (a1) {$x_1$};
    \node[ below = .4cm of a1 ] (a2) {};
    \node[draw,circle, left=.25 of a2] (a2left) {$x_2$};
    \node[draw,circle, right=.25 of a2] (a2right) {$x_2$};
    \node[leaf, below=.35cm of a2      ] (w5) {$1$};

    \draw[<-] (a1) --++(90:2cm) node[right] {$\frac{1}{2}$};
    \draw[e0=20] (a1) edge  node[left,pos=.3] {1} (a2left);
    \draw[e1=20] (a1) edge  node[above right,pos=.6] {1} (a2right);
    \draw[e0=20] (a2left) edge  node[left,pos=.3] {1} (w5);
    \draw[e1=20] (a2left) edge  node[right,pos=.3] {1} (w5);
    \draw[e0=20] (a2right) edge  node[left,pos=.3] {1} (w5);
    \draw[e1=20] (a2right) edge  node[right,pos=.3] {-1} (w5);

    \node[below = .2 of w5] (startedge) {};
    \node[right = 1.6cm of startedge] (finaledge) {};
    \draw[->,dashed] (startedge) -- (finaledge) node[below,pos=.5]{$T_1$};

    \node[draw,circle, right=1.5cm  of a1] (a1) {$x_1$};
    \node[ below = .4cm of a1 ] (a2) {};
    \node[draw,circle, left=.25 of a2] (a2left) {$x_2$};
    \node[draw,circle, right=.25 of a2] (a2right) {$x_2$};
    \node[leaf, below=.35cm of a2      ] (w5) {$1$};

    \draw[<-] (a1) --++(90:2cm) node[right] {$\frac{1}{2}$};
    \draw[e0=20] (a1) edge  node[left,pos=.3] {1} (a2left);
    \draw[e1=20] (a1) edge  node[above right,pos=.6] {$\frac{1}{2}\sqrt{2}+\frac{1}{2}i\sqrt{2}$} (a2right);
    \draw[e0=20] (a2left) edge  node[left,pos=.3] {1} (w5);
    \draw[e1=20] (a2left) edge  node[right,pos=.3] {1} (w5);
    \draw[e0=20] (a2right) edge  node[left,pos=.3] {1} (w5);
    \draw[e1=20] (a2right) edge  node[right,pos=.3] {-1} (w5);

    \node[below = .2 of w5] (startedge) {};
    \node[right = 3.1cm of startedge] (finaledge) {};
    \draw[->,dashed] (startedge) -- (finaledge) node[below,pos=.5]{$H_1$};

    \node[draw,circle, right=3.0cm  of a1] (a1) {$x_1$};
    \node[ below = .4cm of a1 ] (a2) {};
    \node[draw,circle, left=.75 of a2] (a2left) {$x_2$};
    \node[draw,circle, right=.75 of a2] (a2right) {$x_2$};
    \node[leaf, below=.35cm of a2      ] (w5) {$1$};

    \draw[<-] (a1) --++(90:2cm) node[right] {$\frac{1}{4}+\frac{1}{4}\sqrt{2}+\frac{1}{4}i$};
    \draw[e0=20] (a1) edge  node[left,pos=.3] {1} (a2left);
    \draw[e1=20] (a1) edge  node[above right,pos=.6] {$i-i\sqrt{2}$} (a2right);
    \draw[e0=20] (a2left) edge  node[left,pos=.3] {1} (w5);
    \draw[e1=20] (a2left) edge  node[right,pos=.3] {$i-i\sqrt{2}$} (w5);
    \draw[e0=20] (a2right) edge  node[left,pos=.3] {1} (w5);
    \draw[e1=20] (a2right) edge  node[right,pos=.3] {$i+i\sqrt{2}$} (w5);

    \node[below = .2 of w5] (startedge) {};
    \node[right = 3.1cm of startedge] (finaledge) {};
    \draw[->,dashed] (startedge) -- (finaledge) node[below,pos=.5]{$T_2$};

    \node[draw,circle, right=3.0cm  of a1] (a1) {$x_1$};
    \node[ below = .4cm of a1 ] (a2) {};
    \node[draw,circle, left=1.25 of a2] (a2left) {$x_2$};
    \node[draw,circle, right=1.25 of a2] (a2right) {$x_2$};
    \node[leaf, below=.35cm of a2      ] (w5) {$1$};

    \draw[<-] (a1) --++(90:2cm) node[right] {$\frac{1}{4}+\frac{1}{4}\sqrt{2}+\frac{1}{4}i$};
    \draw[e0=20] (a1) edge  node[left,pos=.5] {1} (a2left);
    \draw[e1=20] (a1) edge  node[above right,pos=.6] {$i-i\sqrt{2}$} (a2right);
    \draw[e0=20] (a2left) edge  node[left,pos=.3] {1} (w5);
    \draw[e1=20] (a2left) edge  node[text width = 1.5 cm, right,pos=.5] {$1-\frac{1}{2}\sqrt{2}$\\$-i+\frac{1}{2}i\sqrt{2}$} (w5);
    \draw[e0=20] (a2right) edge  node[left,pos=.5] {1} (w5);
    \draw[e1=20] (a2right) edge  node[text width = 1.5 cm,right,pos=.5] {$-1-\frac{1}{2}\sqrt{2}$\\$+i+\frac{1}{2}i\sqrt{2}$} (w5);

\end{tikzpicture}
    \caption{EVDD example}
    \label{fig:leadingExampleEVDD}
\end{subfigure}
\begin{subfigure}{\textwidth}
    \centering
    \tikzset{every picture/.style={->,thick}}

\begin{tikzpicture}[
    scale=0.3,
    every path/.style={>=latex},
    every node/.style={},
    inner sep=1pt,
    minimum size=0.3cm,
    line width=1pt,
    node distance=.8cm,
    thick,
    font=\scriptsize
    ]

    \node[draw,circle, right=1.6cm  of a1] (a1) {$x_1$};
    \node[draw,circle, below = .4cm of a1 ] (a2) {$x_2$};
    \node[leaf, below=.35cm of a2      ] (w5) {$1$};

    \node[below = .2 of w5] (startedge) {};
    \node[right = .8cm of startedge] (finaledge) {};
    \draw[->,dashed] (startedge) -- (finaledge) node[below,pos=.5]{$H_1$};
    
    \draw[<-] (a1) --++(90:2cm) node[right,pos=.7] {$I\otimes I$};

    \draw[e0=20] (a1) edge  node[left,pos=.3] {$1I$} (a2);
    \draw[e1=20] (a1) edge  node[right,pos=.3] {0} (a2);
    \draw[e0=20] (a2) edge  node[left,pos=.3] {1} (w5);
    \draw[e1=20] (a2) edge  node[right,pos=.3] {0} (w5);

    \node[draw,circle, right=.6cm  of a1] (a1) {$x_1$};
    \node[draw,circle, below = .4cm of a1 ] (a2) {$x_2$};
    \node[leaf, below=.35cm of a2      ] (w5) {$1$};

    \draw[<-] (a1) --++(90:2cm) node[right,pos=.7] {$\frac{1}{\sqrt{2}}I\otimes I$};

    \draw[e0=20] (a1) edge  node[left,pos=.3] {$1I$} (a2);
    \draw[e1=20] (a1) edge  node[right,pos=.3] {$1I$} (a2);
    \draw[e0=20] (a2) edge  node[left,pos=.3] {1} (w5);
    \draw[e1=20] (a2) edge  node[right,pos=.3] {0} (w5);
    
    \node[below = .2 of w5] (startedge) {};
    \node[right = .8cm of startedge] (finaledge) {};
    \draw[->,dashed] (startedge) -- (finaledge) node[below,pos=.5]{$H_2$};

    \node[draw,circle, right=.6cm  of a1] (a1) {$x_1$};
    \node[draw,circle, below = .4cm of a1 ] (a2) {$x_2$};
    \node[leaf, below=.35cm of a2      ] (w5) {$1$};

    \draw[<-] (a1) --++(90:2cm) node[right,pos=.7] {$\frac{1}{2}I\otimes I$};

    \draw[e0=20] (a1) edge  node[left,pos=.3] {$1I$} (a2);
    \draw[e1=20] (a1) edge  node[right,pos=.3] {$1I$} (a2);
    \draw[e0=20] (a2) edge  node[left,pos=.3] {1} (w5);
    \draw[e1=20] (a2) edge  node[right,pos=.3] {1} (w5);

    \node[below = .2 of w5] (startedge) {};
    \node[right = .8cm of startedge] (finaledge) {};
    \draw[->,dashed] (startedge) -- (finaledge) node[below,pos=.5]{$CZ_{1,2}$};

    \node[draw,circle, right=.6cm  of a1] (a1) {$x_1$};
    \node[draw,circle, below = .4cm of a1 ] (a2) {$x_2$};
    \node[leaf, below=.35cm of a2      ] (w5) {$1$};

    \draw[<-] (a1) --++(90:2cm) node[right,pos=.7] {$\frac{1}{2}I\otimes I$};

    \draw[e0=20] (a1) edge  node[left,pos=.3] {$1I$} (a2);
    \draw[e1=20] (a1) edge  node[right,pos=.3] {$1Z$} (a2);
    \draw[e0=20] (a2) edge  node[left,pos=.3] {1} (w5);
    \draw[e1=20] (a2) edge  node[right,pos=.3] {1} (w5);

    \node[below = .2 of w5] (startedge) {};
    \node[right = .8cm of startedge] (finaledge) {};
    \draw[->,dashed] (startedge) -- (finaledge) node[below,pos=.5]{$T_1$};

    \node[draw,circle, right=.6cm  of a1] (a1) {$x_1$};
    \node[draw,circle, below = .4cm of a1 ] (a2) {$x_2$};
    \node[leaf, below=.35cm of a2      ] (w5) {$1$};

    \draw[<-] (a1) --++(90:2cm) node[right,pos=.7] {$\frac{1}{2}I\otimes I$};

    \draw[e0=20] (a1) edge  node[left,pos=.3] {$1I$} (a2);
    \draw[e1=20] (a1) edge  node[right,pos=.3] {$(\frac{1}{2}\sqrt{2}+\frac{1}{2}i\sqrt{2})Z$} (a2);
    \draw[e0=20] (a2) edge  node[left,pos=.3] {1} (w5);
    \draw[e1=20] (a2) edge  node[right,pos=.3] {1} (w5);

    \node[below = .2 of w5] (startedge) {};
    \node[right = 2.1cm of startedge] (finaledge) {};
    \draw[->,dashed] (startedge) -- (finaledge) node[below,pos=.5]{$H_1$};

    \node[draw,circle, right=2.cm  of a1] (a1) {$x_1$};
    \node[draw,circle, below = .4cm of a1 ] (a2) {$x_2$};
    \node[leaf, below=.35cm of a2      ] (w5) {$1$};

    \draw[<-] (a1) --++(90:2cm) node[right,pos=.7] {$(\frac{1}{4}+\frac{1}{4}\sqrt{2}+\frac{1}{4}i)I\otimes I$};

    \draw[e0=20] (a1) edge  node[left,pos=.3] {$1I$} (a2);
    \draw[e1=20] (a1) edge  node[right,pos=.3] {$1X$} (a2);
    \draw[e0=20] (a2) edge  node[left,pos=.3] {1} (w5);
    \draw[e1=20] (a2) edge  node[right,pos=.3] {$i-i\sqrt{2}$} (w5);

    \node[below = .2 of w5] (startedge) {};
    \node[right = 3.1cm of startedge] (finaledge) {};
    \draw[->,dashed] (startedge) -- (finaledge) node[below,pos=.5]{$T_2$};

    \node[draw,circle, right=3.0cm  of a1] (a1) {$x_1$};
    \node[ below = .4cm of a1 ] (a2) {};
    \node[draw,circle, left=1.25 of a2] (a2left) {$x_2$};
    \node[draw,circle, right=1.25 of a2] (a2right) {$x_2$};
    \node[leaf, below=.35cm of a2      ] (w5) {$1$};

    \draw[<-] (a1) --++(90:2cm) node[right] {$(\frac{1}{4}+\frac{1}{4}\sqrt{2}+\frac{1}{4}i)I\otimes I$};
    \draw[e0=20] (a1) edge  node[left,pos=.5] {$1I$} (a2left);
    \draw[e1=20] (a1) edge  node[above right,pos=.6] {$(\frac{1}{2}\sqrt{2}+\frac{1}{2}i\sqrt{2})X$} (a2right);
    \draw[e0=20] (a2left) edge  node[left,pos=.3] {1} (w5);
    \draw[e1=20] (a2left) edge  node[text width = 1.5 cm, right,pos=.5] {$1-\frac{1}{2}\sqrt{2}$\\$-i+\frac{1}{2}i\sqrt{2}$} (w5);
    \draw[e0=20] (a2right) edge  node[left,pos=.5] {1} (w5);
    \draw[e1=20] (a2right) edge  node[text width = 1.5 cm,right,pos=.5] {$-1+\frac{1}{2}\sqrt{2}$\\$-i+\frac{1}{2}i\sqrt{2}$} (w5);
\end{tikzpicture}
    \caption{LIMDD example}
    \label{fig:leadingExampleLIMDD}
\end{subfigure}
\caption{Example of \autoref{thm:edgelabelbound}, \autoref{cor:LIMDD-width_mainText} and \autoref{cor:EVDD-bound}, showing the intermediate EVDD and LIMDDs of the 2-qubit Clifford$+T$ circuit from \Cref{fig:circuit_example}.\\
\autoref{thm:edgelabelbound}: All edge values in the bulk and the magnitude of the root edge value are contained in $R_{2n+2t+1}$ after $t$ T-gates have been applied. (Note that the magnitude of the root edge label for the two rightmost states is $|\frac{1}{4}+\frac{1}{4}\sqrt{2}+\frac{1}{4}i|^2=\frac{1}{4}+\frac{1}{8}\sqrt{2}$, which is in $R_3$.) \\
\autoref{cor:LIMDD-width_mainText}: The EVDD width is at most $2^{\min({\#H,2\cdot\#CZ + \#T})}$ where $\#G$ is the number of applied $G$-gates for $G\in\{H,CZ,T\}$.\\
\autoref{cor:EVDD-bound}: The LIMDD width is at most $2^t$ after $t$ T-gates have been applied.}
\label{fig:leadingExample}
\end{figure}

We subtly change the way a decision diagram is stored, as the theorem does not give a guarantee on the size of the root edge label. We choose to not store the root edge label $r$ itself, but $|r|^2$ which by the theorem can be described in linear space. Fortunately, this suffices for a full description of $\ket{\phi}$ because quantum states are defined up to a global phase, i.e., for any $\lambda \in \mathbb{C}$ with $|\lambda| = 1$, the states $\ket{\phi}$ and $\lambda \cdot \ket{\phi}$ are indistinguishable~\cite{nielsen2010quantum}.

\subsection{Simulation with decision diagrams needs only short numbers}
\label{subsec:measurement_probabilities_are_in_qisq}
We now prove \Cref{thm:edgelabelbound-dynamic}, an extension of \Cref{thm:edgelabelbound} that shows that not only the edge values of the output state of a Clifford+$t\times T$ circuit $U$ are short, but in fact any edge value encountered when simulating $U$ using the gate-simulation algorithms for EVDD~\cite{lai1994evbdd,tafertshofer1994factored,tafertshofer1997factored} or LIMDD~\cite{vinkhuijzen2023limdd}.
What's more, the complex numbers encountered when simulating computational-basis measurement also remain short.

Note that the proof requires a careful analysis of the simulation algorithms for both EVDD and LIMDD. We analyze the simulation algorithms that are provided by~\cite{vinkhuijzen2023limdd}. The details of the analysis are in \autoref{app:quantum_simulation_coefficients_analysis}.

\begin{theorem}[DD-simulation needs only short numbers.]%
\label{thm:edgelabelbound-dynamic}
Let $U$ be an $n$-qubit Clifford+$t\times T$ unitary.
Using an EVDD-low or LIMDD-low representation of the quantum state and only storing the magnitude $|r|^2$ of the root edge label $r$, the complex numbers encountered when simulating $U$ gate-by-gate on $\ket{0}^{\otimes n}$ are elements of $R_{O(n+t)}$ and can thus each be described in $O(n+t)$ bits.
The same holds for simulating a computational-basis measurement on the top-qubit of $U\ket{0}^{\otimes n}$.
\end{theorem}

\begin{proof}[Proof sketch]%

\textit{Gates:} As the gates $H, CZ,S$ and $T$ span the entire Clifford$+T$ gate set, we limit our analysis to those gates.

We know by~\autoref{thm:edgelabelbound} that the bulk edge labels and magnitude of the root edge label are in $R_{O(n+t)}$ before and after the application of each gate. We only need to analyze which complex numbers are encountered during the execution of the algorithms for applying gates.

By \autoref{lemma:basic_operations_in_qisq_are_tractable}, we know that a constant number of arithmetic operations of numbers in $R_{O(n+t)}$ gives a number in $R_{O(n+t)}$ again. Therefore, it is enough to show that a constant number of arithmetic operations are performed on the edge values of the input DD, or that intermediate values of the calculation equal an edge label of the DD after the application of a gate. 

The only numerical calculations in the gate-application algorithms of the gates $S,CZ$, and $T$ are multiplication of the edge values of the initial DD with $\pm1,i$ or $e^{i\pi /4}$ (the nonzero entries of the $S,T$, and $CZ$ gates). Such multiplication of an initial edge value, which is in $R_{O(n+t)}$, is again in $R_{O(n+t)}$. For EVDD, this concludes the proof.
For LIMDD, normalization needs to ensure high determinism~\cite{vinkhuijzen2023limdd}, leading to new normalization of edge labels. This can involve one of the following: replace high edge label $a$ by $\frac{1}{a}$; multiplication of the root edge label by the high edge label; multiplication of high edge label and/or root edge label by $-1,i,-i$ (\Cref{lemma:AnalysisMakeEdge}). The updated high edge label is part of the decision diagram after application of the gate, so it is again in $R_{O(n+t)}$. Since a product of subsequent edge labels along a path in the decision diagram is in $R_{O(n+t)}$ (\Cref{lem:simple_sums_quotients_of_subseq_edgelabels_are_small}), this also holds for the updated root edge label. So we conclude that the gates $S,CZ$, and $T$ can be applied only using numbers in $R_{O(n+t)}$. The figure below shows an example of applying a $T$-gate to qubit $k$. In this example, the edges are swapped if the multiplicative inverse of the high edge label is smaller than the high edge label itself, ensuring high determinism (cf. \Cref{alg:make-edge,alg:get-labels}). %

\begin{figure}[h]
\centering
\tikzset{every picture/.style={->,thick}}

\begin{tikzpicture}[
    scale=0.3,
    every path/.style={>=latex},
    every node/.style={},
    inner sep=1pt,
    minimum size=0.3cm,
    line width=1pt,
    node distance=.8cm,
    thick,
    font=\scriptsize
    ]

    \node[draw,circle] (a1) {$x_k$};
    \node[draw,circle, below = .4cm of a1, xshift=-.65cm] (a2) {$x_{k+1}$};
    \node[draw,circle, below = .4cm of a1, xshift= .65cm] (a3) {$x'_{k+1}$};
    
    \draw[<-] (a1) --++(90:2cm) node[right,pos=.7] {$c$};
    \draw[e0 = 0] (a1) edge  node[above left] {$1$} (a2);
    \draw[e1 = 0] (a1) edge  node[above right] {$a$} (a3);

    \node[draw,circle,right= 4cm of a1] (a1a) {$x_k$};
    \node[draw,circle, below = .4cm of a1a, xshift=-.65cm] (a2) {$x_{k+1}$};
    \node[draw,circle, below = .4cm of a1a, xshift= .65cm] (a3) {$x'_{k+1}$};
    
    \draw[<-] (a1a) --++(90:2cm) node[right,pos=.7] {$c$};
    \draw[e0 = 0] (a1a) edge  node[above left] {$1$} (a2);
    \draw[e1 = 0] (a1a) edge  node[above right] {$a\cdot e^{i\pi/4}$} (a3);

    \draw[->,dashed, shorten >=1cm,  shorten <=1cm] (a1) edge node[above] {apply $T$ gate}  (a1a);

    \node[draw,circle,right= 5cm of a1a] (a1b) {$x_k$};
    \node[draw,circle, below = .4cm of a1b, xshift=-.65cm] (a2b) {$x'_{k+1}$};
    \node[draw,circle, below = .4cm of a1b, xshift= .65cm] (a3b) {$x_{k+1}$};
    
    \draw[<-] (a1b) --++(90:2cm) node[above right,pos=.7] {$c\cdot a\cdot e^{i\pi/4} X$};
    \draw[e0 = 0] (a1b) edge  node[above left] {$1$} (a2b);
    \draw[e1 = 0] (a1b) edge  node[above right] {$\frac{1}{a\cdot e^{i\pi/4}}$} (a3b);

    \draw[->,dashed, shorten >=1cm,  shorten <=1cm] (a1a) edge node[below] {if $a\cdot e^{i\pi/4}>\frac{1}{a\cdot e^{i\pi/4}}$} node[above] {normalize node $x_k$} (a1b);
    
\end{tikzpicture}
\end{figure}

DD updates by $H$-gates may need more numerical calculations, because edges need to be added and subtracted. In the appendix (\Cref{cor:coeff_bound_during_simulation--appendix}), we unravel the structure of all intermediate values in terms of the edge values in the input DD. As all intermediate values are always a product or sum of a constant number of edge values of the input DD, the intermediate values during the numerical calculations are also tractable in size: they remain all in $R_{O(n+t)}$.

\textit{Measurement:} Measuring the top-qubit requires a recursive calculation of the squared norm of the underlying nodes. In the appendix (\Cref{cor:measurements_with_efficient_numbers}), we show, by using pre- and postconditions on the squared norm algorithm, that all numerical values during this recursion are in $R_{O(n+t)}$.
\end{proof}

Finally, we note that although our proof shows that each \emph{intermediate} value during the measurement algorithm is in $R_{O(n+t)}$, we can make a stronger claim for the \emph{final} return value, i.e., the measurement outcome probability: this probability is in $Q_{2n,2n+t}$.
To see this, note that the probability of outcome 0 (similar for 1) can be written as $\sum_{x\in\{0,1\}^{n-1}}|\langle 0\ x\mid \phi\rangle|^2=\sum_{x\in\{0,1\}^{n-1}} \langle 0 \ x\mid\phi\rangle\langle\phi\mid 0 \ x\rangle$, and is thus is a sum of $2^{n-1}$ (diagonal) entry in the density matrix, which are all in $Q_{n,t}$ by \autoref{lemma:vector-entry-lemma}, up to a factor $2^n$. Since $\{2^{-n} \cdot x \mid x \in Q_{n, t} \} \subseteq Q_{n, 2n+t}$ and $\{\sum_{i=1}^{2^{n-1}}a_i\mid a_i\in Q_{n, 2n+t}\}\subseteq Q_{2n,2n+t}$, the probability of outcome 0 and 1 are in $Q_{2n,2n+t}$.

\section{Bounds on the size of decision diagrams}
\label{sec:numnodes_bound}

In the previous section, we showed that the edge coefficients in a decision diagram can be represented efficiently in the $T$-count of a circuit. 
This section shows that the number of nodes of a LIMDD is also tractable in the $T$-count of a circuit, and for EVDD fixed-parameter tractable in the minimum of $H$-, $CZ$-, and $T$-count (see \Cref{thm:ddsizebounds}, stated below). The theorem is independent of the canonicity rule that is used and are valid for both the `low' and `L2' canonicity rules.

\begin{theorem}[DD size bounds]%
\label{thm:ddsizebounds}
Let $U$ be an $n$-qubit unitary, consisting of only $X,Y,Z,H, S, T, SWAP$ and $CZ$ gates.
Then $U\ket{0}^{\otimes n}$ is represented by a LIMDD resp. EVDD with the number of nodes upper bounded by $n\cdot 2^{\#T}$ resp. $n\cdot 2^{\min(\#H,2\cdot \#CZ + \#T)}$ (for both the `low' and `L2' DD variants).
Here, $\#G$ means the number of $G$ gates in the circuit  $U$ for $G=H, T,CZ$ gates.
\end{theorem}

Our proof, worked out in detail in \Cref{appendix:numnodes_bounds} and sketched for LIMDD in \Cref{subsec:bounds_on_LIMDD_size} and for EVDD in \Cref{subsec:bounds_on_EVDD_size}, depends on the size of the stabilizer group of the represented quantum state, quantified using the notion of the (local) stabilizer nullity.

\subsection{Bounds on LIMDD size}
\label{subsec:bounds_on_LIMDD_size}

The main ingredient for the LIMDD size bound proof is the stabilizer nullity. This quantity depends on the Pauli stabilizers of a state, as defined in~\cref{sec:qc}.

\begin{definition}[stabilizer nullity~\cite{beverland2020lower}]
    The stabilizer nullity $\nu(\ket{\phi})$ of an $n$-qubit quantum state $\ket{\phi}$ the difference between the number of qubits and the number of generating stabilizers of $\ket{\phi}$. Formally, $\nu(\ket{\phi}):= n - \log_2(|S(\ket{\phi})|)$.
\end{definition}

For an $n$ qubit state, the stabilizer nullity is an integer between 0 and $n$. Examples are $\nu(\ket{\phi}) = 0$ when $\ket{\phi}$ is a stabilizer state, and $\nu((T\cdot H \ket{0})^{\otimes n})) = n$.

We denote by the $T$-count of an $n$-qubit state $\ket{\phi}$, the smallest $T$-count over all Clifford+$k\times T$ circuits which output $\ket{\phi}$ on input $\ket{0}^{\otimes n}$.
In general, the stabilizer nullity is upper bounded by the $T$-count, as shown in the following lemma.

\begin{lemma}
\label{lemma:stabilizer-count_t-count-mainText}
    For every $\ket{\phi}\in\cktstates^{n}_t$ the stabilizer nullity is upper bounded by $\ket{\phi}$'s $T$-count. Stated differently, for every $\ket{\phi}\in\cktstates^{n}_t$ we have $\nu(\ket{\phi})\leq t$.
\end{lemma}
\begin{proof}[Proof sketch (full proof in \Cref{lemma:CliffTstabilizers})]%
    The idea of the proof is that every $T$-gate can remove at most one generating Pauli string stabilizer. This is because only stabilizers with an $X$ or $Y$ are not preserved under conjugation: $TIT^\dagger = I$,$TZT^\dagger = Z$, but $TXT^\dagger = \frac{X+Y}{\sqrt{2}}$ and $TYT^\dagger = \frac{Y-X}{\sqrt{2}}$. By elimination of the stabilizers, there can be at most one such unpreserved generating stabilizer per applied $T$-gate.
\end{proof}

Now, we relate the stabilizer nullity to the LIMDD width.

\begin{theorem}
\label{thm:LIMDD-width-mainText}
    The LIMDD width of a quantum state is at most exponentially in the stabilizer nullity. In other words, for every quantum state $\ket{\phi}$ the LIMDD width is at most $2^{\nu(\ket{\phi})}$.
\end{theorem}
\todo{Optional: add figure to clarify proof}
\begin{proof}[Proof sketch (full proof in \cref{thm:m_stabs_LIMDD_width})]%
    The proof goes in two parts, which are proven in detail in \cref{lemma:splitCase} resp. \cref{thm:m_stabs_LIMDD_width}.
    
    First, we note that every node $v$ in a LIMDD can be regarded as an (unnormalized) quantum state and therefore has its own Pauli string stabilizers, which we write as $S(v)$. The number of stabilizers of a LIMDD node can only increase with respect to its child if this node has only one child node. Moreover, in case of a single child, the number of generating stabilizers can only increase by one. In summary, we know that for every LIMDD node $v$ with (possibly equal) childs $w_0$ and $w_1$ we have either (1) $w_0=w_1$ and $\log_2(|S(v)|)=\log_2(|S(w_0)|)+1$, or (2) $|S(v)|=\min\{|S(w_0)|,|S(w_1)|\}$ (see \cref{lemma:splitCase}).

    Second, we note that the root node, which represents the state $\ket{\phi}$, has $n-\nu(\ket{\phi})$ generating stabilizers by definition of $\nu(\ket{\phi})$. Because the terminal nodes have no stabilizers, there are at least $n-\nu(\ket{\phi})$ pairs of adjacent layers of LIMDD nodes where the number of stabilizers must increase. By the first part of the proof, we know that an increase of stabilizers between two layers implies that all nodes in that layer can have at most one child. Therefore, there are at most $\nu(\ket{\phi})$ layers where nodes can have more than one child. Hence, as any LIMDD node has at most two childs, the LIMDD width is at most $2^{\nu(\ket{\phi})}$. (A formal proof using induction can be found in \cref{thm:m_stabs_LIMDD_width}.)
\end{proof}

As we related the stabilizer nullity to both the $T$-count and the LIMDD width, we can now prove the tractability of the LIMDD size in terms of the $T$-count. This is the main result of this section.

\begin{corollary}
\label{cor:LIMDD-width_mainText}
    Every $n$ qubit quantum circuit consisting of Clifford gates and t $T$-gates with initial state $\ket{0}^{\otimes n}$ can be simulated with a LIMDD of width at most $2^t$, i.e., at most $n\cdot 2^t$ nodes.
\end{corollary}
\begin{proof}
    This follows directly from \Cref{lemma:stabilizer-count_t-count-mainText} and \Cref{thm:LIMDD-width-mainText}, and the observation that a EVDD has at most $n$ layers.
\end{proof}

\autoref{fig:leadingExampleLIMDD} contains an example of a simulated quantum circuit that obeys the bound on the LIMDD width and size.
We note that the bound on the LIMDD size that we provide is not tight. For example the magic state $(T\cdot H\ket{0})^{\otimes n}$ has stabilizer nullity $n$, and hence needs $n$ T-gates to prepare (by \Cref{lemma:stabilizer-count_t-count-mainText}), but can be represented by a LIMDD of $n$ nodes.

Note that an stronger version of \Cref{lemma:stabilizer-count_t-count-mainText} exists for Toffoli gates. Every Toffoli gate reduces the stabilizer nullity by at most 3, while a Toffoli gate implemented in the Clifford+$T$ gate set requires at least 7 $T$ gates~\cite{gosset2013algorithmtcount}. The details of the proof are worked out in~\Cref{lemma:CliffToffolistabilizers}. With this lemma, \Cref{cor:LIMDD-width_mainText} can be generalized as follows.
\begin{corollary}
\label{cor:LIMDD-width-with-Toffoli_mainText}
    Every $n$ qubit quantum circuit consisting of Clifford gates, t $T$-gates and t' Toffoli gates with initial state $\ket{0}^{\otimes n}$ can be simulated with a LIMDD of width at most $2^{t+3t'}$, i.e., at most $n\cdot 2^{t+3t'}$ nodes.
\end{corollary}

Finally, we note that all proofs in this section only rely on the property of the $T$ gate that a $T$ gate can remove at most one generating Pauli stabilizer, which is also satisfied by any rotation gate $R_Z(\theta):=\left(\begin{smallmatrix} 1&0\\
0&e^{i\theta}
\end{smallmatrix}\right)$ for some $\theta \in [0, 2\pi)$. Hence

\begin{corollary}
\label{cor:LIMDD_width_rotZ-gates}
    \Cref{cor:LIMDD-width_mainText} and \Cref{cor:LIMDD-width-with-Toffoli_mainText} still hold if each individual $T$ gate is replaced by some $Z$ rotation.
\end{corollary}

\subsection{Bounds on EVDD size}
\label{subsec:bounds_on_EVDD_size}

In this section, we show that the size of a EVDD is tractable in the minimum of the $H$-, $CZ$, and $T$-count.
The proof is very similar to the proof in the previous section for bounds on the LIMDD size. The main difference is a different notion of stabilizer nullity, using a subset of Pauli stabilizers from~\cref{sec:qc}.

\begin{definition}[local stabilizer nullity]
    The local stabilizer nullity $\nu_{L}(\ket{\psi})$ of an $n$-qubit quantum state $\ket{\psi}$ is $n$ minus the number of generating Pauli stabilizers of that state that are of the form $\pm P_k = \pm I^{\otimes k-1}\otimes P \otimes I^{\otimes n-k}$ where $P\in \{X, Y, Z\}$. We call a Pauli stabilizer of this form a local stabilizer.
\end{definition}

The local stabilizer nullity of an $n$ qubit quantum state is an integer between $0$ and $n$. 
Examples are $\nu_L(\ket{\phi}) = 0$ when $\ket{\phi}$ is a computational basis state, and $\nu_L(\ket{GHZ_n})=n$ for $\ket{GHZ_n}=\frac{1}{\sqrt{2}}\left(\ket{0}^{\otimes n} + \ket{1}^{\otimes n}\right)$.

We now show a general upper bound to the local stabilizer nullity in terms of applied gates.

\begin{lemma}
\label{lemma:singlet_stabilizer_count-mainText}
    Fix a positive integer $n$ and let $U$ be a $n$-qubit Clifford+$k\times T$ unitary, so that it can be implemented by a circuit $C$ containing only single qubit Clifford gates, SWAP, CZ-, and T-gates.
    Then $\nu_L(U\ket{0}^{\otimes n}) \leq \min({\#H,2\cdot\#CZ + \#T})$, where $\#G$ denotes the number of $G$ gates in $C$.
\end{lemma}
\begin{proof}[Proof sketch (full proof in \cref{lemma:h-cz-t-count_EVDD})]%
    The idea of the proof is that the $\ket{0}^{\otimes n}$ state has $n$ stabilizers $Z_j$. Gate updates for all Clifford+$T$ gates except $H$ are invariant to $Z_j$ up to sign: $CZ_jC = \pm Z_j$ for $C\in\{X,Y,Z,S,T,CZ,SWAP\}$. An $H$ gate can turn one $Z_j$ stabilizer into an $X_j$ stabilizer. (\Cref{lemma:h-count_EVDD})
    
    Gate updates for Clifford+$T$ gates are again invariant to $X_j$ up to sign, except for $H$, $CZ$ and $T$. If one resp. two $X_j$ stabilizers are present when an $T$ resp. $CZ$ gate is applied, the $X_j$ stabilizer can be turned into non-local stabilizer.

    Thus, to remove a local stabilizer, at least an $H$ gate, and a $CZ$ or $T$ gate is needed. The $CZ$ can destroy two local stabilizers at once, provided that two $H$ gates were applied. So the local stabilizer nullity is at most $m = \min({\#H,2\cdot\#CZ + \#T})$.
\end{proof}

Now, we relate the local stabilizer nullity to the EVDD width.

\begin{theorem}
\label{thm:EVDD_width-mainText}
    The EVDD width of a quantum state $\ket{\phi} \in \cktstates^n$ is at most $2^{\nu_L(\ket{\phi})}$.
\end{theorem}
\begin{proof}[Proof sketch (full proof in \Cref{theorem:EVDD_width_from_number_of_stabilizers})]%
    The idea of the proof is very similar to the proof of \Cref{thm:LIMDD-width-mainText}. Every local stabilizer $\pm P_k$ prevents the EVDD to increase the number of nodes between two adjacent layers. Therefore, the number of nodes can increase only at $m=\nu_L(\ket{\phi})$ layers. Thus the width of any EVDD layer is at most $2^{m}$.
\end{proof}

As we related the local stabilizer nullity to the EVDD width and the $H$-, $CZ$- and $T$-count, we can now prove the tractability of the EVDD size in terms of the $H$-, $CZ$- and $T$-count. This results in the following Corollary.%

\begin{corollary}
\label{cor:EVDD-bound}
    Every $n$ qubit quantum circuit consisting of single qubit Clifford gates, SWAP, CZ-,and T-gates with initial state $\ket{0}^{\otimes n}$ can be simulated with a EVDD of width at most $2^m$, i.e., at most $n\cdot 2^m$ nodes for $m=\min({\#H,2\cdot\#CZ + \#T})$.
\end{corollary}
\begin{proof}
    This follows directly from \Cref{lemma:singlet_stabilizer_count-mainText} and \Cref{thm:EVDD_width-mainText}, and the observation that a EVDD has at most $n$ layers.
\end{proof}

\autoref{fig:leadingExampleEVDD} contains an example of a simulated quantum circuit that obeys the bound on the EVDD width and size.
One should note that the bound on the EVDD size that we provide is not tight. For example the $\ket{GHZ_n}$ state has local stabilizer nullity $n$, and hence needs $n$ $H$-gates, and at least $n/2$ $CZ$-gates or $T$-gates to prepare (by \Cref{lemma:singlet_stabilizer_count-mainText}), but can be represented by a EVDD of $2n$ nodes.

\section{Experiments\label{sec:experiments}}

\vspace{-.5em}
In this section, we present our experimental evaluations. We implemented the symbolic coefficients in the decision diagram tool Q-Sylvan~\cite{brand2025q}. Our implementation is open source.\footnote{\url{https://github.com/Quist-A/q-sylvan}}
We compare the performance of quantum circuit simulation with decision diagrams using either of two ways to store the complex numbers in the decision diagram: using the symbolic parameterization described in \Cref{sec:coeff-bounds} (denoted {\algebraic}) with `low' normalization or as floating point numbers with `L2' normalization (denoted \float).\footnote{We only compare \algebraic-low and \float-L2, but not \algebraic-L2 and \float-low. For \float, `L2' has been empirically observed to be more stable than 'low'~\cite{brand2025q}. For \algebraic, the `L2' normalization does not exist: `L2' normalization requires taking square roots, which is in general not possible in the symbolic parameterization from \Cref{sec:coeff-bounds}.}

We benchmark three types of circuits: Grover's algorithm, W-state preparation and random circuits. These benchmarks are chosen as they allow exact representation using Clifford+$T$ circuits, contrary to e.g. the Quantum Fourier Transform.
Moreover, the $n$-qubit W-state is a suitable test case for the symbolic tool, as it has a high $T$-count~\cite{arunachalam2022parameterized}, but $O(n)$ EVDD width in the perfect accurate case. We hypothesize that the intermediate EVDDs are large, and rounding errors lead to merging of non-equivalent nodes for the float EVDD, returning incorrect quantum states. 
For the random circuit benchmarks, we used circuits with gates randomly sampled from the Clifford+$T$ gates with 20 to 100 qubits and circuit depth between 700 and 1000. 
Each circuit is followed by a single-qubit measurement on the top qubit.
Our benchmarks were run on a single core AMD Ryzen 9 7900X Processor and 64 GB memory. Each benchmark was run for at most 30 minutes.

\subsubsection*{Results.}
We provide the benchmark results %
on {\float} vs. {\algebraic} with runtime (\Cref{fig:qisq-vs-float_plots_runtime}), final number of nodes (\Cref{fig:qisq-vs-float_plots_final_nodes}) and peak nodes (\Cref{fig:qisq-vs-float_plots_peak_nodes}).

Comparing {\algebraic} vs. {\float}, we surprisingly observe that {\algebraic} is often faster. Moreover, we sometimes observe that {\float} returns an incorrect (i.e., more than 5\% different from {\algebraic}) measurement outcome probability.

\begin{figure}[h]
\centering
\begin{subfigure}{0.32\textwidth}
    \includegraphics[width=\textwidth]{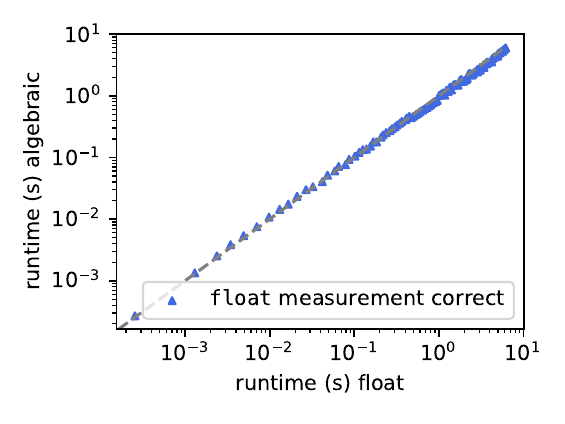}
    \caption{Grover}
\end{subfigure}
\begin{subfigure}{0.32\textwidth}
    \includegraphics[width=\textwidth]{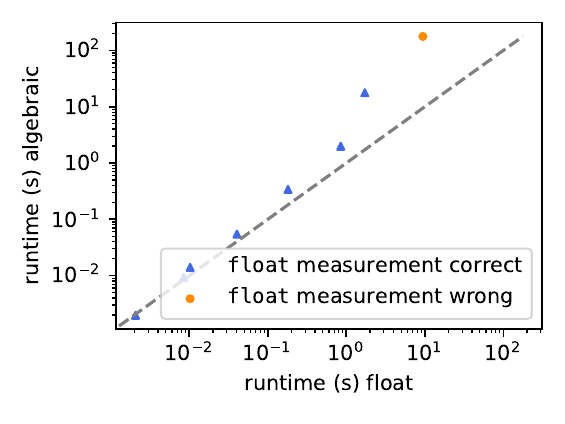}
    \caption{W-state}
\end{subfigure}
\begin{subfigure}{0.32\textwidth}
    \includegraphics[width=\textwidth]{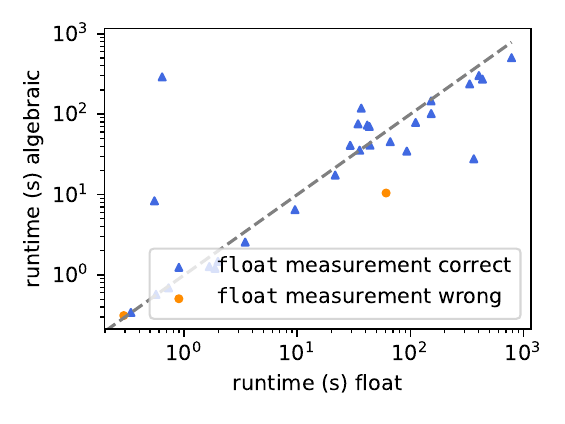}
    \caption{Random circuits}
\end{subfigure}

\caption{Runtime benchmark results on {\float} vs. {\algebraic}.}
\label{fig:qisq-vs-float_plots_runtime}
\end{figure}

\begin{figure}[h]
\centering
\begin{subfigure}{0.32\textwidth}
    \includegraphics[width=\textwidth]{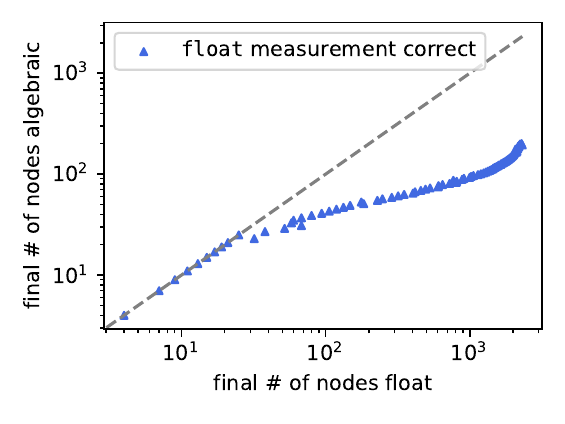}
    \caption{Grover}
\end{subfigure}
\begin{subfigure}{0.32\textwidth}
    \includegraphics[width=\textwidth]{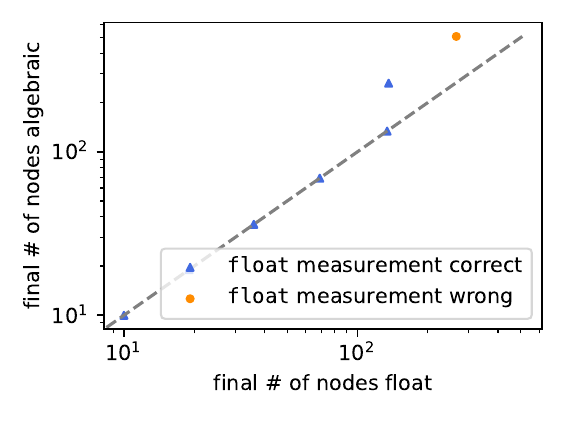}
    \caption{W-state}
\end{subfigure}
\begin{subfigure}{0.32\textwidth}
    \includegraphics[width=\textwidth]{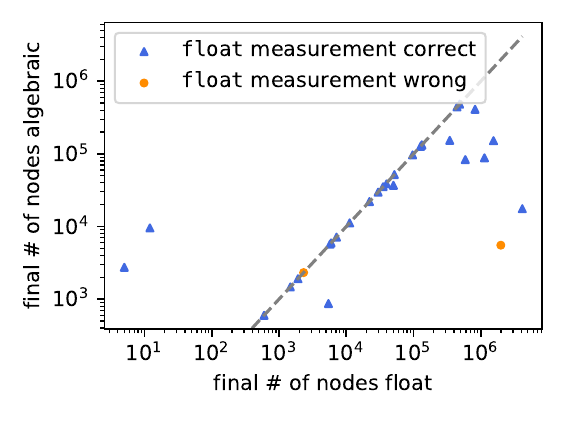}
    \caption{Random circuits}
\end{subfigure}

\caption{Final nodes benchmark results on {\float} vs. {\algebraic}.}
\label{fig:qisq-vs-float_plots_final_nodes}
\end{figure}

\begin{figure}[h]
\centering
\begin{subfigure}{0.32\textwidth}
    \includegraphics[width=\textwidth]{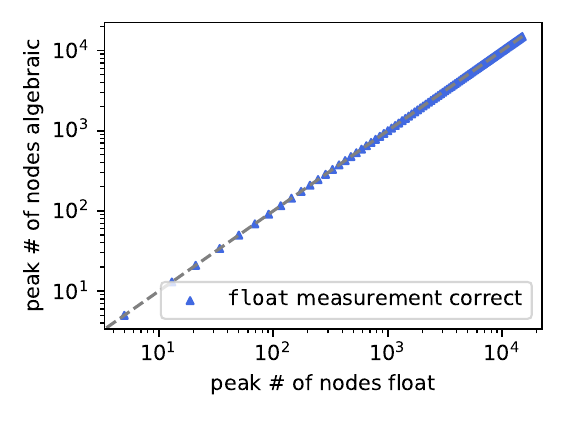}
    \caption{Grover}
\end{subfigure}
\begin{subfigure}{0.32\textwidth}
    \includegraphics[width=\textwidth]{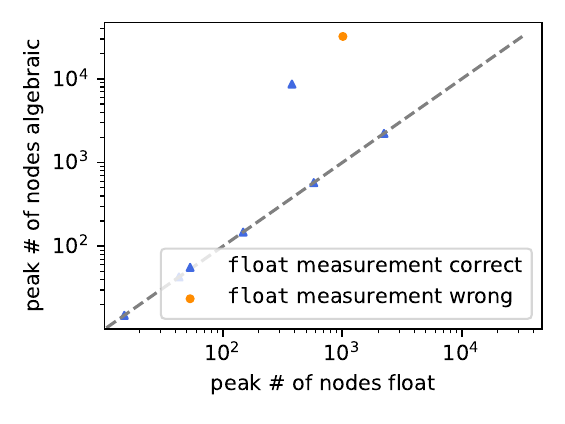}
    \caption{W-state}
\end{subfigure}
\begin{subfigure}{0.32\textwidth}
    \includegraphics[width=\textwidth]{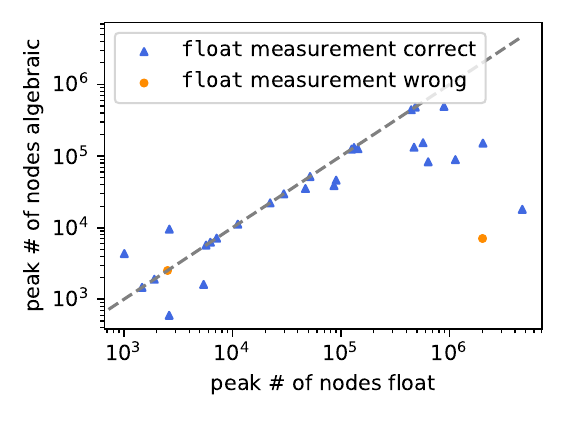}
    \caption{Random circuits}
\end{subfigure}

\caption{Peak nodes benchmark results on {\float} vs. {\algebraic}.}
\label{fig:qisq-vs-float_plots_peak_nodes}
\vspace{-1em}
\end{figure}

From the {\float} vs. {\algebraic} benchmarks in \Cref{fig:qisq-vs-float_plots_runtime}, \ref{fig:qisq-vs-float_plots_final_nodes} and \ref{fig:qisq-vs-float_plots_peak_nodes}, we see that for large Grover circuits {\algebraic} is faster and has fewer final nodes. As the peak of number of nodes is equal for {\algebraic} and {\float}, {\float} is not able to merge equivalent nodes after the peak due to numerical errors.\\
For the W-state, we observe the opposite. The large W-state circuits for {\float} are faster than {\algebraic} and use fewer nodes. As a downside, {\float} provides an incorrect decision diagram which give an wrong measurement probability. \\
For random circuits, the number of nodes for {\algebraic} is often less than for {\float}. Interestingly, these additional nodes do not always lead to stability for {\float}, as some large benchmarks are unstable for {\float}. The runtimes of {\float} and {\algebraic} differ a lot: neither is consistently faster.

\section{Conclusion}
\label{sec:conclusion}
We have shown that exact simulation of Clifford+$T$ quantum circuits using LIMDD decision diagrams can be made fixed-parameter tractable in the number of $T$ gates, while we prove a weaker statement for EVDD decision diagrams.
Our result consists of two aspects: first, we proved that the size of our algebraic parameterization of the complex numbers encountered in the decision diagram is linearly bounded in the $T$-count.
This is the first theoretical guarantee on the required edge label size for exact reasoning using decision diagrams.
Second, we proved that the number of nodes in the LIMDD is exponentially bounded in the $T$-count and constant in the number of Clifford gates, based on a novel connection between stabilizer nullity and number of decision diagram nodes.
The combined results imply that LIMDD simulation runtime is fixed-parameter tractable in the number of $T$ gates.
Our EVDD proofs follow a similar reasoning.
Our scaling statements also materialize in the proof-of-principle EVDD implementation, which outperforms the conventional floating-point-based EVDD in the scenarios studied, providing provable reliable measurement outcomes that do not suffer from floating point inaccuracies.

In the future, our proofs can be extended to different gate sets, such as Clifford+(fixed) rotations, and beyond simulation, e.g. to circuit equivalence checking~\cite{thanos2023fast}, verification~\cite{Abdulla2026verification,chen2025autoq,chen2025automata,mcmillan1992symbolic} and synthesis~\cite{clarke1981design,zulehner2017one,zulehner2017improving,abdollahi2006analysis,zak2025reducing}.
Our proof techniques ---tailoring the decision diagram's edge label parameterization as well as the stabilizer group size--- also open up new opportunities to numerical stability and decision diagram size bounds for different applications than quantum circuit analysis.

\paragraph{Acknowledgements}
This publication is part of the project Divide \& Quantum  (with project number 1389.20.241) of the research program NWA-ORC which is (partly) financed by the Dutch Research Council (NWO).
The second
author acknowledges the support received through the NWO Quantum Technology
program (project number NGF.1582.22.035).

\bibliographystyle{plain}
\bibliography{bibliography.bib,bib2.bib}

\appendix

\section{Bounds on size of edge label coefficients\label{appendix:coeffs-bound}}

\subsection{Characterizing all possible edge label coefficients in a DD}
\label{subsec:characterize_possible_edge_labels}

\begin{definition}\label{def:cp}
    Let $\rho=\frac{1}{2^n}\sum_{P\in \mathcal{P}_n}c_P P$ be a density matrix in the Pauli basis of an $n-$qubit state. The set $C^\rho=\{c_P\mid P\in\mathcal{P}_n\}$ is called the Pauli spectrum of $\rho$~\cite{beverland2020lower}. We define the set $C_t$ as the union over all Pauli spectra of density matrices $\rho=\dyad{\phi}$ with $\ket{\phi}\in\cktstates^{n}_t$ for any $n\in\mathbb{N}$.
\end{definition}

\begin{lemma}
\label{lem:C_k-representation}
     In \autoref{def:cp}, every $c_P\in C_t$ can be written as $c_P=0$ or $\pm1$ if $t=0$, and for $t>0$ $c_P$ is of the form $c_P=\frac{\ell}{\sqrt{2^t}}+\frac{m}{\sqrt{2^{t-1}}}$ with $\ell\in\{-2^t,-2^t+1,\dots,-1,0,1,\dots,2^t-1,2^t\}$ and $m\in\{-2^{t-1},-2^{t-1}+1,\dots,-1,0,1,\dots,2^{t-1}-1,2^{t-1}\}$.
\end{lemma}
\begin{proof}
    For $t=0$ this statement follows directly from the fact that every stabilizer state can be written as a sum of its stabilizers~\cite{aaronson2004improved}. So $C_0=\{0,1,-1\}$.

    Now, we show the lemma with induction to $t$. 

    Assume we apply a Clifford gate $C$ to the state $\frac{1}{2^n}\sum_{P \in \mathcal{P}_n}c_P P$. As $\mathcal{P}_n$ is closed under conjugation with Clifford gates, the resulting state $\frac{1}{2^n}\sum_{P \in \mathcal{P}_n}c_P CPC^{\dagger}$ has the same coefficients $c_P$.

    Assume we apply a $T$-gate to the state $\frac{1}{2^n}\sum_{P = P_1 \otimes P_{2\dots n} \in \mathcal{P}_n}c_P P$. Without loss of generality, we apply the T-gate to the first qubit. Then, the resulting state is 
    \begin{align*}
    &\frac{1}{2^n}\sum_{P = P_1 \otimes P_{2\dots n} \in \mathcal{P}_n}c_P TP_1T^{\dagger}\otimes P_{2\dots n} \\
    &= \frac{1}{2^n}\sum_{P_{2\dots n} \in \mathcal{P}_{n-1}}c_{I\otimes P_{2\dots n}} TIT^{\dagger}\otimes P_{2\dots n} + c_{X\otimes P_{2\dots n}} TXT^{\dagger}\otimes P_{2\dots n} 
    + c_{Y\otimes P_{2\dots n}} TYT^{\dagger}\otimes P_{2\dots n} + c_{Z\otimes P_{2\dots n}} TZT^{\dagger}\otimes P_{2\dots n}\\
    &= \frac{1}{2^n}\sum_{P_{2\dots n} \in \mathcal{P}_{n-1}}c_{I\otimes P_{2\dots n}} I\otimes P_{2\dots n} + c_{X\otimes P_{2\dots n}} \frac{X+Y}{\sqrt{2}}\otimes P_{2\dots n} 
    + c_{Y\otimes P_{2\dots n}} \frac{Y-X}{\sqrt{2}}\otimes P_{2\dots n} + c_{Z\otimes P_{2\dots n}} Z \otimes P_{2\dots n}\\
    &= \frac{1}{2^n}\sum_{P_{2\dots n} \in \mathcal{P}_{n-1}}c_{I\otimes P_{2\dots n}} I\otimes P_{2\dots n} + \frac{c_{X\otimes P_{2\dots n}} - c_{Y\otimes P_{2\dots n}}}{\sqrt{2}} X\otimes P_{2\dots n} 
    + \frac{c_{X\otimes P_{2\dots n}} + c_{Y\otimes P_{2\dots n}}}{\sqrt{2}} Y\otimes P_{2\dots n} + c_{Z\otimes P_{2\dots n}} Z \otimes P_{2\dots n}\\
    &= \frac{1}{2^n}\sum_{P = P_1 \otimes P_{2\dots n} \in \mathcal{P}_n}c'_P P_1 \otimes P_{2\dots n}.
    \end{align*}
    
    Here, we used the equalities:
    \begin{align*}
        TIT^{\dagger} = I,\qquad
        TXT^{\dagger} = \frac{X+Y}{\sqrt{2}},\qquad
        TYT^{\dagger} = \frac{Y-X}{\sqrt{2}},\qquad
        TZT^{\dagger} = Z.
    \end{align*}
    We see that:
    \[c'_{P_1\otimes P_{2\dots n}} = \begin{cases}
        c_{P_1\otimes P_{2\dots n}} & \text{if } P_1=I,Z\\
        \frac{1}{\sqrt{2}}(c_{X\otimes P_{2\dots n}} - c_{Y\otimes P_{2\dots n}}) & \text{if } P_1=X\\
        \frac{1}{\sqrt{2}}(c_{X\otimes P_{2\dots n}} + c_{Y\otimes P_{2\dots n}}) & \text{if } P_1=Y
    \end{cases}\]

    Hence we see that, without loss of generality,\footnote{Note that $\frac{1}{\sqrt{2}}(c_P-d_P) = \frac{1}{\sqrt{2}}(c_P+d^-_P)$ for $d^-_P = -d_P\in C_t$, as by the induction hypothesis $d_P\in C_t$ implies $-d_P\in C_t$.} every element $c'_P\in C_{t+1}$ can be written as $c'_P = c_P$ or $c'_P = \frac{1}{\sqrt{2}}(c_P+d_P)$ with $c_P,d_P\in C_t$. By induction we know that $c_P = \frac{\ell}{\sqrt{2^t}}+\frac{m}{\sqrt{2^{t-1}}}$ and $d_P = \frac{\ell'}{\sqrt{2^t}}+\frac{m'}{\sqrt{2^{t-1}}}$ with $\ell,\ell'\in\{-2^t,\dots,2^t\}$ and $m,m'\in\{-2^{t-1},\dots,2^{t-1}\}$.\\
    If $c'_P = c_P$, it is enough to show $C_t\subseteq C_{t+1}$. This follows directly by $$c'_P = c_P = \frac{\ell}{\sqrt{2^t}}+\frac{m}{\sqrt{2^{t-1}}} = \frac{2m}{\sqrt{2^{t+1}}} + \frac{\ell}{\sqrt{2^t}},\quad
    \text{where }2m\in\{-2^{t+1},\dots,2^{t+1}\}.$$
    \\
    If $c'_P = \frac{1}{\sqrt{2}}(c_P+d_P)$, then we have $$c'_P = \frac{\ell}{\sqrt{2^{t+1}}}+\frac{m}{\sqrt{2^{t}}} + \frac{\ell'}{\sqrt{2^{t+1}}}+\frac{m'}{\sqrt{2^{t}}} = \frac{\ell+\ell'}{\sqrt{2^{t+1}}} + \frac{m+m'}{\sqrt{2^{t}}},$$
    with $\ell+\ell'\in\{-2^{t+1},\dots, 2^{t+1}\}$ and $m+m'\in\{-2^{t},\dots,2^{t}\}$.

    This shows the induction.
\end{proof}

Without loss of generality, we use the normalization rule in EVDD/LIMDD that the low edge always has edge coefficient 0 or 1 (see \autoref{sec:qdd}). If the low edge has coefficient 0, then the high edge has coefficient 1. Using this normalization of edges, we show the following lemma.

\begin{lemma}
\label{lem:DDlabel=fraction-of-C_k}
    Every edge coefficient in the bulk of EVDD-low and LIMDD-low  
    representing $f:\{0,1\}^n\to D$ with $D\subseteq \co$ equals a quotient of two elements in $D=\{f(x) \mid x \in \{0, 1\}^n \}$ (for LIMDD: up to a factor $-1$, $i$ or $-i$) where the denominator of the quotient is nonzero.
\end{lemma}

\begin{proof}
    
    Let $(v)$ be a node in the decision diagram with outgoing edges to $(v_0)$ with label 1 and $(v_1)$ with label $a$. \\
    Consider a path $p$ from the root node to $(v)$. We call the product of all edge weights on this path $\alpha$.
    From $(v)$ we can take a path $p_0$ via $(v_0)$ to the terminal node with only 1's on the edges, as every node has an outgoing edge with label 1 due to the normalization rule. Similarly, we can take the path $p_1$ from $(v)$ via $(v_1)$ to the terminal node with only 1's on the edges, except for the edge between $(v_0)$ and $(v_1)$ (which has label $a$).
    Now, $\alpha$ is the product of the edge labels in $p|p_0$ and we call $\beta$ be the product of the edge labels in $p|p_1$.\footnote{The notation $p|q$ for two paths means that the paths are glued together.}\\
    We see that $\alpha,\beta\in D$. This holds because every path in a decision diagram corresponds to an assignment $x\in\{0,1\}^n$. The products of the edge labels in this path correspond to the value $f(x)$, which is indeed in $D$.\\
    We also see that $\frac{\beta}{\alpha}=a$ as $p|p_0$ and $p|p_1$ have the same edge labels except for the one between $(v)$ and $(v_0)$ (label 1) and the one between $(v)$ and $(v_1)$ (label $a$).\\
    Note that the denominator $\alpha$ is nonzero, as the path $p|p_0$ contains no zero: if $p$ contains an edge label zero, then node $(v)$ would represent the zero node and could hence be removed from the diagram; $p_0$ only contains only edge labels 1.\\
    As the node $(v)$ was chosen arbitrary, we see that every edge label $a$ can be expressed as $a=\frac{\beta}{\alpha}$ for $\alpha,\beta\in D$ with $\alpha\not=0$.

    For LIMDD, a small extension of the proof is needed. When following a path $x$ along a LIMDD, the product of the edge labels along $x$ does not necessary equal $f(x)$. Due to LIMs, which have all entries in $0,1,-1,i,-i$, the value of the product of the edge labels equals $i^kf(x)$ for some $k\in\{0,1,2,3\}$. Thus, every edge label $a$ can be expressed as $a=i^k\frac{\beta}{\alpha}$ for $\alpha,\beta\in D$ with $\alpha\not=0$ and $k\in\{0,1,2,3\}$.
    
    \autoref{fig:quotient-of-statevectorentries} contains a picture to give intuition for the proof of this lemma.

\end{proof}

\begin{lemma}
\label{lemma:ratio_of_state-vector_entries_is_linearly_representable}
    Let $a,b\in Q_{n,t}$ with the guarantee that $b$ is a real number. Then the number $\frac{a}{b}$ is $R_{2(n+t)+1}$. Therefore, the quotient of two state vector entries (with nonzero denominator) is in $R_{2(n+t)+1}$.
    
\end{lemma}
\begin{proof}

    We can write $a = \frac{\ell}{\sqrt{2^t}}+\frac{m}{\sqrt{2^{t-1}}} + i \left(\frac{\ell'}{\sqrt{2^t}}+\frac{m'}{\sqrt{2^{t-1}}}\right)$ and $b = \frac{\ell''}{\sqrt{2^t}}+\frac{m''}{\sqrt{2^{t-1}}}$ with $\ell,\ell',\ell''\in\{-2^{n+t},-2^{n+t}+1,\dots,2^{n+t}-1,2^{n+t}\}$ and $m,m',m''\in\{-2^{n+t-1},-2^{n+t-1}+1,\dots,2^{n+t-1}-1,2^{n+t-1}\}$. \\
    Then by direct calculation we see:
    \begin{align*}
        \frac{a}{b}&=\frac{\frac{\ell}{\sqrt{2^t}}+\frac{m}{\sqrt{2^{t-1}}} + i \left(\frac{\ell'}{\sqrt{2^t}}+\frac{m'}{\sqrt{2^{t-1}}}\right)}{\frac{\ell''}{\sqrt{2^t}}+\frac{m''}{\sqrt{2^{t-1}}}}\\
        &=\frac{\ell+m\sqrt{2}+i(\ell'+m'\sqrt{2})}{\ell''+m''\sqrt{2}}\\
        &=\frac{\ell+m\sqrt{2}+i(\ell'+m'\sqrt{2})}{\ell''+m''\sqrt{2}}\cdot\frac{\ell''-m''\sqrt{2}}{\ell''-m''\sqrt{2}}\\
        &= \frac{\ell\ell''-2mm''}{\ell''^2-2m''^2}+\frac{m\ell''-\ell m''}{\ell''^2-2m''^2}\sqrt{2}+\frac{\ell'\ell''-2m'm''}{\ell''^2-2m''^2}i+\frac{m'\ell''-\ell'm''}{\ell''^2-2m''^2}i\sqrt{2}.
    \end{align*}

    Hence, we conclude $\frac{a}{b}\in R_{2m+1}$.

    The quotient of two state vector entries can be written as $\frac{\langle x\mid\phi\rangle}{\langle y\mid\phi\rangle}=\frac{\langle x\mid\phi\rangle\langle \phi\mid y\rangle}{\langle y\mid\phi\rangle\langle \phi\mid y\rangle} = \frac{2^n\cdot\langle x\mid \rho\mid y\rangle}{2^n\cdot\langle y\mid \rho\mid y\rangle}$. By \autoref{lemma:vector-entry-lemma} we know that the numerator and denominator are in $Q_{n,t}$ with a real denominator. Hence, the analysis above can be applied, directly proving that the quotient is in $R_{2(n+t)+1}$.
\end{proof}

\subsection{Quantum simulation with decision diagrams using efficient coefficient data structure}
\label{app:quantum_simulation_coefficients_analysis}

In this section, we prove the statement from \Cref{thm:edgelabelbound-dynamic} by tracking all numerical values in the gate update algorithms for decision diagrams. The gate update algorithms we discuss and analyze here are those for LIMDDs. We use the algorithms that are provided by~\cite{vinkhuijzen2023limdd}, reproducing them below as: \textsc{Add} (\Cref{alg:add-limdds}), and \textsc{Follow} (\Cref{alg:follow}), \textsc{MakeEdge} (\Cref{alg:make-edge}), and \textsc{GetLabels}  (\Cref{alg:get-labels}). Algorithms for EVDDs are similar, as all Pauli-LIMs on edge values of a LIMDD can be set to the all-identity Pauli string $I^{\otimes m}$ for an EVDD.

For convenience, we define the numerical parts of the edge labels in the decision diagram as $\alpha_{x_1x_2\dots x_m}$ for the edge pointing to the node that is reached by the path $x_1x_2\dots x_m$ starting from the root node.

In the following lemmas, we show preconditions and postconditions on functions, starting from the subroutines of gate update algorithms.

\begin{lemma}
\label{lemma:AnalysisGetLabels}
    The algorithm $\textsc{GetLabels}$ called on a LIM with numerical coefficient $\lambda$ only computes (and hence also returns) numerical values of the form $i^\ell\lambda^{(-1)^k}$ for $\ell\in\{0,1,2,3\}, k\in\{0,1\}$. Moreover, the numerical value of $\rootlim$ is always of the form $i^\ell(\lambda)^k$ for $\ell\in\{0,1,2,3\}, k\in\{0,1\}$.
\end{lemma}
\begin{proof}
    The $\textsc{GetLabels}$ algorithm contains only numerical calculations in lines 5, 8, 9, and 10. The products $g_0Pg_1$ always have a numerical factor $i^\ell$ as they are all Pauli strings. Therefore, the numerical values on these lines are of the form $i^\ell\lambda^{(-1)^k}$.\\
    The numerical value of $\rootlim$ is calculated in line 10. By direct inspection, we see that it is of the form $i^\ell(\lambda)^k$.
\end{proof}

\begin{lemma}
\label{lemma:AnalysisMakeEdge}
    The algorithm $\makeedge$ called on edges $\ledge{A}{v_0}$ and $\ledge{B}{v_1}$ where $A,B$ have numerical coefficients $a,b$ only computes numerical values of the form $i^\ell a$, $i^\ell b$ or $i^\ell\left(\frac{a}{b}\right)^{(-1)^k}$ for $\ell\in\{0,1,2,3\}, k\in\{0,1\}$. It returns an edge of the form $\ledge{A'}{v'}$ where the numerical label of $A'$ is of the form $i^\ell a$ or $i^\ell b$.
\end{lemma}
\begin{proof}
    The first base case (line 2) only swaps the input edges. As the values $i^\ell\left(\frac{a}{b}\right)^{(-1)^k}$ and $i^\ell\left(\frac{b}{a}\right)^{(-1)^{1-k}}$ are the same, this does not change the analysis.

    The second base case (line 4) only involves trivial numerical calculations.

    The final case is interesting. First, in line 9, $\hat{A}=A^{-1}B$ is calculated, which gives a numerical value of the form $i^\ell\left(\frac{b}{a}\right)$. The factor $i^\ell$ is due to the multiplication of the Pauli strings $A^{-1}$ and $B$.\\
    Then, in line 10, $\textsc{GetLabels}$ is called on this numerical value $\hat{A}$. By \Cref{lemma:AnalysisGetLabels} this takes only numerical values that are of the form $i^\ell\left(\frac{a}{b}\right)^{(-1)^k}$ for $\ell\in\{0,1,2,3\}, k\in\{0,1\}$. Also, the numerical label of $\rootlim$ is $i^\ell(\frac{b}{a})^k$.\\
    Finally, in line 12, $\rootlim$ is multiplied by $A$ resulting in the value $i^\ell a(\frac{b}{a})^k = i^\ell a^{1-k}b^k$. This is of the form $i^\ell a$ or $i^\ell b$. This new $\rootlim$ is also the edge label $A'$ of the returned edge $\ledge{A'}{v'}$.
\end{proof}

\begin{lemma}
\label{lemma:AnalysisFollow}
    The algorithm $\textsc{Follow}$ called on edge $\ledge{A}{v}$ and value $j$, where LIM $A$ has numerical value $\lambda$ and $v$ is reachable by path $x_1,x_2,\dots,x_m$ returns edge $\ledge{A'}{v'}$ where $A'$ has numerical label $i^\ell \lambda\cdot\alpha_{x_1,\dots,x_{m+1}}$ and $v'$ can be reached by the path $x_1,\dots,x_{m+1}$. Here, the Boolean value $x_{m+1}$ is in $\{0,1\}$.
\end{lemma}
\begin{proof}
    Line 3 of the $\textsc{Follow}$ algorithm shows that $\gamma$ equals $i^\ell \lambda$. As we have $A = P_n\otimes\dots\otimes P_1$, we see that $y_n=j$ if $P_n$ is diagonal and $y_n=1-j$ if $P_n$ is anti-diagonal, distinguishing the if-statement in line 4. \\
    The recursive call on $\textsc{Follow}$ in lines 5 resp. 7 directly returns the low edge resp. high edge of $v$, which have numerical edge label $\alpha_{x_1,\dots,x_{m+1}}$. So the returned edge $\ledge{A'}{v'}$ has numerical label of the form $i^\ell \lambda\cdot\alpha_{x_1,\dots,x_{m+1}}$.
\end{proof}

\begin{lemma}
\label{lemma:AnalysisAdd}
    The algorithm $\textsc{Add}$ called on edges $\ledge{A}{v}$ and $\ledge{B}{w}$ where $A,B$ have numerical coefficients $a,b$ and $v$ is reachable by path $x_1,x_2,\dots,x_m$ and $w$ is reachable by path $y_1,y_2,\dots,y_m$ returns an edge $\ledge{C}{v'}$ where $C$ has numerical coefficient $c$. The value $c$ equals $a \cdot i^\ell\prod_{d=m+1}^n \alpha_{x_1,x_2,\dots,x_d} + b \cdot i^{\ell'}\prod_{d=m+1}^n \alpha_{y_1,y_2,\dots,y_d}$, where $x_1,x_2,\dots,x_n$ and $y_1,y_2,\dots,y_n$ are extensions of the paths $x_1,x_2,\dots,x_m$ and $y_1,y_2,\dots,y_m$ from $v$ resp. $w$ to the terminal node.
    \sloppy All numerical values computed during the call of $\textsc{Add}$ are of the form $p,q,p+q,\frac{p}{q},\frac{p+q}{r+s}$, where $p,q,r,s$ are in $\left\{a \cdot i^\ell\prod_{d=m+1}^{m''} \alpha_{x_1,x_2,\dots,x_d}, i^\ell\prod_{d=m'}^{m''} \alpha_{x_1,x_2,\dots,x_d}, b \cdot i^{\ell}\prod_{d=m+1}^{m''} \alpha_{y_1,y_2,\dots,y_d},i^{\ell}\prod_{d=m'}^{m''} \alpha_{y_1,y_2,\dots,y_d}\right\}$ with $m',m''\geq m$.
\end{lemma}
\begin{proof}
    For the base case (line 2), the statement is immediate.

    As the algorithm is symmetric, the swap of $v$ and $w$ in line 3 does not influence the analysis.

    In the base case when $m=n$ the statement is immediate.\\
    The remainder of the proof is via induction on $m$. A call to $\textsc{Add}$ on edges $\ledge{A}{v}$ and $\ledge{B}{w}$ where $A,B$ have numerical coefficients $a,b$ and $v$ is reachable by path $x_1,x_2,\dots,x_{m-1}$ and $w$ is reachable by path $y_1,y_2,\dots,y_{m-1}$ goes as follows:\\
    In line 4, the value $C$ is $i^l \frac{b}{a}$.\\
    In lines 6 and 7, an edge is returned with numerical edge label $i^\ell\prod_{d=m}^n \alpha_{x_1,x_2,\dots,x_d} + \frac{b}{a}i^{\ell'}\prod_{d=m}^n \alpha_{y_1,y_2,\dots,y_d}$, as we know by induction and \Cref{lemma:AnalysisFollow}. \\
    Line 8 gives an edge with a similar edge label by \Cref{lemma:AnalysisMakeEdge}: only the values of $\ell$ and $\ell'$ may change. \\
    Line 9 multiplies the edge label by $a$, which gives the new edge label of the edge that is returned: $a\cdot i^\ell\prod_{d=m}^n \alpha_{x_1,x_2,\dots,x_d} + b\cdot i^{\ell'}\prod_{d=m}^n \alpha_{y_1,y_2,\dots,y_d}$.

    The specific form of all computed numerical values during the calculation is similarly proved with induction and the use of \Cref{lemma:AnalysisMakeEdge}.
\end{proof}

\begin{lemma}
\label{lemma:AnalysisHGate}
    The algorithm $\textsc{HGate}$ called on edge $\ledge{A}{v}$ reachable by path $x_1,x_2,\dots,x_m$, where LIM $A$ has numerical value $a=\alpha_{x_1,x_2,\dots,x_m}$ returns an edge with numerical LIM value $\frac{1}{\sqrt{2}}\left[i^\ell\prod_{d=m}^n \alpha_{x_1\dots x_d} + i^{\ell'}\prod_{d=m}^n \alpha_{y_1\dots y_d}\right]$ such that $x_i=y_i$ for all $i\in\{1,\dots,k\}$. 
    Moreover, it computes only numerical values of the form $p,q,p+q,\frac{p}{q},\frac{p+q}{r+s},\frac{1}{\sqrt{2}}(p+q)$ with $p,q,r,s$ in $\left\{a \cdot i^\ell\prod_{d=m+1}^{m''} \alpha_{x_1,x_2,\dots,x_d}, i^\ell\prod_{d=m'}^{m''} \alpha_{x_1,x_2,\dots,x_d}, a \cdot i^{\ell}\prod_{d=m+1}^{m''} \alpha_{y_1,y_2,\dots,y_d},i^{\ell}\prod_{d=m'}^{m''} \alpha_{y_1,y_2,\dots,y_d}\right\}$ with $m',m''\geq m$.
\end{lemma}
\begin{proof}
    If $m$ equals $k$ (the level where the $H$-gate will be applied), then two calls to the $\textsc{Add}$ algorithm are done in line 4. By \Cref{lemma:AnalysisAdd} the numerical values during the computation are contained in the expressions above in the statement of this lemma. Again by \Cref{lemma:AnalysisAdd}, the returned edges have numerical edge values of the form $i^\ell\prod_{d=k+1}^n \alpha_{x_1\dots x_d} + i^{\ell'}\prod_{d=k+1}^n \alpha_{y_1\dots y_d}$, with $x_i=y_i$ for all $i\in\{1,\dots,k\}$. \\
    The \makeedge call returns again an edge with such edge value.\\
    Finally, in line 4, this value is multiplied by $\frac{1}{\sqrt{2}}$, which gives the value $\frac{1}{\sqrt{2}} \left[i^\ell\prod_{d=k+1}^n \alpha_{x_1\dots x_d} + i^{\ell'}\prod_{d=k+1}^n \alpha_{y_1\dots y_d}\right]$.\\
    \sloppy The multiplication by $a=\alpha_{x_1,x_2,\dots,x_m}$ in line 7 then finally returns an edge with numerical value $\frac{1}{\sqrt{2}}\left[i^\ell\prod_{d=m}^n \alpha_{x_1\dots x_d} + i^{\ell'}\prod_{d=m}^n \alpha_{y_1\dots y_d}\right]$.

    For the case when $m$ is smaller than $k$, two recursive calls to the $H$-gate algorithm are made in line 6. By induction, the values during this computation are again of the form described above in the statement of this lemma. 
\end{proof}

\begin{lemma}
\label{lemma:basic_operations_in_qisq_are_tractable}
    The result of an addition, multiplication or division of two values of the form $a+b\sqrt{2}+ci+di\sqrt{2}$ with $a,b,c,d$ fractions of the form $\frac{p}{q}$ with $p,q\in\{-2^{k},\dots,0,\dots,2^{k}\}$ is of the form $a+b\sqrt{2}+ci+di\sqrt{2}$ with $a,b,c,d$ fractions of the form $\frac{p}{q}$ with $p,q\in\{-2^{ck},\dots,0,\dots,2^{ck}\}$ for some constant $c$. In other words, the size of the output of an addition, multiplication, or division is linear in the size of the inputs.
\end{lemma}
\begin{proof}
    In \cref{alg:Qisq2-manipulations} we show the basic arithmetic operations in this data structure. We see that every operation can be applied efficiently, i.e., with a constant number of addition and multiplication operations on fractions. Therefore, the size of the output for every addition or multiplication of fractions is linear in the size of the inputs.
\end{proof}

\begin{algorithm}
\caption{Algorithms for arithmetic operations in the data structure $a+b\sqrt{2}+ci+di\sqrt{2}$, where $a,b,c,d$ are fractions.}
\label{alg:Qisq2-manipulations}
\begin{algorithmic}[1]
    \Procedure{add}{x,y}
    \State result.init()
    \State result.a = add\_frac(x.a,y.a)
    \State result.b = add\_frac(x.b,y.b)
    \State result.c = add\_frac(x.c,y.c)
    \State result.d = add\_frac(x.d,y.d)
    \State result.reduce() \Comment{turn result.a, result.b,... into unique fractions by Euclids algorithm}
    \State \textbf{return} result
    \EndProcedure
\end{algorithmic}

\begin{algorithmic}[1]
    \Procedure{subtract}{x,y}
    \State result.init()
    \State result.a = subtract\_frac(x.a,y.a)
    \State result.b = subtract\_frac(x.b,y.b)
    \State result.c = subtract\_frac(x.c,y.c)
    \State result.d = subtract\_frac(x.d,y.d)
    \State result.reduce() \Comment{turn result.a, result.b,... into unique fractions by Euclids algorithm}
    \State \textbf{return} result
    \EndProcedure
\end{algorithmic}

\begin{algorithmic}[1]
    \Procedure{multiply}{x,y}
    \State result.init()
    \State result.a = add\_frac(multiply\_frac(x.a,y.a),2$\cdot$ multiply\_frac(x.b,y.b),-1$\cdot $multiply\_frac(x.c,y.c),-2$\cdot $multiply\_frac(x.d,y.d))
    \State result.b = add\_frac(multiply\_frac(x.a,y.b),-1$\cdot$ multiply\_frac(x.c,y.d),multiply\_frac(x.b,y.a),-1$\cdot$ multiply\_frac(x.d,y.c))
    \State result.c = add\_frac(multiply\_frac(x.a,y.c),2$\cdot$ multiply\_frac(x.b,y.d),multiply\_frac(x.c,y.a),2$\cdot$ multiply\_frac(x.d,y.b))
    \State result.d = add\_frac(multiply\_frac(x.a,y.d), multiply\_frac(x.d,y.a), multiply\_frac(x.b,y.c), multiply\_frac(x.c,y.b))
    \State result.reduce() \Comment{turn result.a, result.b,... into unique fractions by Euclids algorithm}
    \State \textbf{return} result
    \EndProcedure
\end{algorithmic}

\begin{algorithmic}[1]
    \Procedure{divide}{x,y}
    \State result.init()
    \State temp.init()
    \State temp.a,temp.b,temp.c,temp.d := y.a,y.b,-y.c,-y.d
    \State x = multiply(x,temp)
    \State y = multiply(y,temp) \Comment{y is now of the form $a+b\sqrt{2}$}
    \State temp.a,temp.b,temp.c,temp.d := y.a,-y.b,0,0
    \State x = multiply(x,temp)
    \State y = multiply(y,temp) \Comment{y is now of the form $a=\frac{p}{q}$}
    \State temp = 1/y \Comment{1/y is fraction $\frac{q}{p}$}
    \State result = multiply(x,temp)
    \State temp.clear()
    \State result.reduce() \Comment{turn result.a, result.b,... into unique fractions by Euclids algorithm}
    \State \textbf{return} result
    \EndProcedure
\end{algorithmic}
\end{algorithm}

\begin{lemma}
\label{lem:simple_sums_quotients_of_subseq_edgelabels_are_small}
    The values $p,q,p+q,\frac{p}{q},\frac{p+q}{r+s},\frac{1}{\sqrt{2}}(p+q)$ with $p,q,r,s$ of the form $ i^\ell\prod_{d=m'}^{m''} \alpha_{x_1,x_2,\dots,x_d}$ are in $R_{c(n+k)}$ for some constant $c$.
\end{lemma}
\begin{proof}
    Every number of the form $i^\ell\prod_{d=m'}^{m''} \alpha_{x_1,x_2,\dots,x_d}$ is in $R_{2(n+k)+1}$. The argument for this goes similarly to the argument in \Cref{thm:edgelabelbound}: every product of subsequent edge labels distinguishes between two paths with only `1' as edge label (similar to \Cref{lem:DDlabel=fraction-of-C_k}), such that this value is a fraction of two state vector entries (for LIMDD: up to multiplication with $-1$, $i$ or $-i$). Such fraction is in $R_{2(n+k)+1}$ by \Cref{lemma:ratio_of_state-vector_entries_is_linearly_representable}.

    Moreover, by \Cref{lemma:basic_operations_in_qisq_are_tractable} any constant number of additions, multiplications or divisions applied to numbers in $R_{2(n+k)+1}$ is in $R_{c(2(n+k)+1))}$. This proves the lemma.
\end{proof}

\begin{corollary}
\label{cor:coeff_bound_during_simulation--appendix}
    Every numerical value during the execution of a gate algorithm on an EVDD or LIMDD is in $R_{O(n+k)}$, where $n$ is the number of qubits and $k$ is the $T$-count of the state the gate algorithm is applied to.
\end{corollary}

\begin{algorithm}[h!]
	\begin{algorithmic}[1]
        \Procedure{HGate}{$\Edge~ \ledge{A}{v}$ \textbf{with} $A\in \pauli$-LIM,
        	$k \in \set{1,\dots, \index(v)}$}
		\If{$v \notin \textsc{HGate}\cache$}\label{sg:cache1} \Comment{Compute result once for $v$ and store in cache:}
		\If{$\index(v)=k$}
		\State $\textsc{HGate}\cache[v] := \frac{1}{\sqrt 2} \cdot\makeedge(\textsc{Add}(\low v,  \high v), \textsc{Add}(\low v,  -\high v))$\hspace{-1em}
		\label{hg:hk}
		\Else
		\State $\textsc{HGate}\cache[v] := \makeedge(\textsf{HGate}(\low v,k), \textsf{HGate}(\high v,k))$
		\label{hg:rebuild}
		\EndIf
		\EndIf
		\State \Return $H_k A H_k^\dagger  \cdot \textsc{HGate}\cache[v]$
							\Comment{Retrieve result from cache}
		\label{hg:rebuild2}
		\EndProcedure
	\end{algorithmic}
    \caption{(Reproduced from~\cite{vinkhuijzen2023limdd}) Apply gate $H$ to qubit $k$ for \pauli-\limdd \ledge Av. We let $n = \index(v)$.}
	\label{alg:H-gate}
\end{algorithm}

\begin{algorithm}[h!]
	\begin{algorithmic}[1]
		\Procedure{Add}{\Edge $\ledge[e=] Av$, \Edge $\ledge[f=] Bw$
		 \textbf{with} $\index(v) = \index(w)$}
		\If{$\index(v)=0$}
			 \Return $\leafedge{A+B}$ \Comment{$A,B \in \mathbb C$}
		\EndIf
			\If{$v\not\beforeq w$} \Return $\textsc{Add}(\ledge Bw, \ledge Av)$\Comment{Normalize for cache lookup}
			\EndIf
			\label{algline:add:add-swap-cache-lookup}
			\State $C:=\rootlabel(\ledge {A^{-1}B}w)$
			\label{algline:add:add-cache-factor-C}
			\If{$(v, C, w)\notin \textsc{Add-Cache} $}\Comment{Compute result for the first time:}
			\label{algline:add:lookup-cache-factor-C}
		\State \Edge $a_0:=\textsc{Add}(\follow 0{\ledge{}{v}}, \follow 0{\ledge{C}{w}})$
			\label{algline:add:add-0}
		\State \Edge $a_1:=\textsc{Add}(\follow 1{\ledge{}{v}}, \follow 1{\ledge{C}{w}})$
			\label{algline:add:add-1}
		\State $\textsc{Add-Cache}[(v, C, w)] :=  \makeedge(a_0, a_1)$\Comment{Store in cache}
			\label{algline:add:add-makeedge}
			\label{algline:add:add-store}
		\EndIf
		\State\Return $A\cdot \textsc{Add-Cache}[(v,C,w)]$\Comment{Retrieve from cache}
		\EndProcedure
	\end{algorithmic}
	\caption{(Reproduced from~\cite{vinkhuijzen2023limdd}) Given two $n$-LIMDD edges $e,f$, constructs a new LIMDD edge $a$ with $\ket{a}=\ket{e}+\ket{f}$.
	}
	\label{alg:add-limdds}
	\label{alg:add-limdds-cache}
\end{algorithm}

\begin{algorithm}
	\begin{algorithmic}[1]
		\Procedure{Follow}{\Edge $\ledge[e]{\lambda P_n \otimes \dots \otimes P_1}{v}$, $x_n, \dots, x_k \in \set{0,1}$ \textbf{with} $n = \index(v)$ and $k \geq 1$}
        \If{$k > n$} \Return $\ledge{\lambda P_n \otimes \dots \otimes P_1}{v} $
        \Comment{End of bit string}
        \EndIf
		\State $\gamma \bra{y_n \dots y_k} := \bra{x_n \dots x_k} \lambda P_n \otimes \dots \otimes P_k$ \Comment{$O(n)$-computable LIM operation}
		\label{algline:amplitude:lim}
		\If{$y_n = 0$}\Comment{$y_n = 0$}
		\State \Return $\gamma \cdot \textsc{Follow}(\textsf{low}_v, y_{n-1}, \dots, y_k)$
		\Else\Comment{$y_n = 1$}
		\State \Return $\gamma \cdot \textsc{Follow}(\textsf{high}_v, y_{n-1}, \dots, y_k)$
		\EndIf
		\EndProcedure
	\end{algorithmic}
	\caption{(Reproduced from~\cite{vinkhuijzen2023limdd}) \textsc{Follow}: a generalization of \textsc{ReadAmplitude}, returning edges.}
	\label{alg:follow}
\end{algorithm}

\begin{algorithm}
	\begin{algorithmic}[1]
        \Procedure{MakeEdge}{\Edge $\ledge[e_0]{A}{v_0}$, $\ledge[e_1]{B}{v_1}$, \textbf{with}  $v_0, v_1$ reduced, $A \neq 0$ \textbf{or} $B \neq 0$}
            \If{$v_0\not\beforeq v_1$ \textbf{or} $A=0$} \Comment{Enforce \textbf{low precedence} and enable \textbf{factoring}}
                \State \Return\label{l:swap}
                        $(X \otimes \id^{\otimes n})\cdot \text{MakeEdge}(e_1, e_0)$
			\EndIf
            \If{$B = 0$}%
            \State $v_1 := v_0$     \Comment{Enforce  \textbf{zero edges}}
            \State $v := \lnode{\id^{\otimes n}}{v_0}{0}{v_0}$\label{l:low1}
            \Comment{Enforce \textbf{low factoring}}
            \State $\rootlim := \id \otimes A$   \Comment{$\rootlim \ket v = \ket 0 \otimes 
                                        A \ket{v_0} + \ket 1 \otimes  B \ket{v_1}$}\label{l:low2}
            \Else
            \State $\hat A := A^{-1}B$ \Comment{Enforce \textbf{low factoring}}\label{l:low3}
            \State $\highlim, \rootlim:=\textsc{GetLabels}(\hat A,v_0,v_1)$
            \label{algline:makeedge-get-labels}
            \Comment{Enforce \textbf{high determinism}}
            \State $v := \lnode{\id^{\otimes n}}{v_0}{\highlim}{v_1}$ 
             \Comment{$\rootlim \ket v = \ket 0 \otimes 
                                         \ket{v_0} + \ket 1 \otimes A^{-1} B \ket{v_1}$}
                \label{l:low4}
            \State $\rootlim := (\id \otimes A) \rootlim $ 
             \Comment{$(\id \otimes A)\rootlim \ket v = \ket 0 \otimes  A \ket{v_0} + \ket 1 \otimes  B \ket{v_1}$}
            \EndIf
			\State $v_{\text{r}}:=$ Find or create unique table entry $\unique[v] = (v_0, \highlim, v_1)$       
			\label{algline:find-v-in-unique}
            \Comment{Enforce \textbf{merge}}
            \State \Return $\ledge{\rootlim}{v_{\text{r}}}$
		\EndProcedure
	\end{algorithmic}
	\caption{(Reproduced from~\cite{vinkhuijzen2023limdd}) 
        Algorithm \makeedge takes two root edges to (already reduced) nodes $v_0,v_1$, the children of a new node, and returns a reduced node with a root edge.
	It assumes that %
	$\index(v_0) = \index(v_1) = n$.
    The bold comments refer to reductions defined in~\cite{vinkhuijzen2023limdd}.
}
	\label{alg:make-edge}
\end{algorithm}

\begin{algorithm}
	\caption{(Reproduced from~\cite{vinkhuijzen2023limdd}) 
		Algorithm for finding LIMs $\highlim$ and $\rootlim$ required by \makeedge.
		Its parameters represent a semi-reduced node $\lnode[v]{\id}{v_0}{\lambda P}{v_1}$
		and it returns LIMs $\highlim, \rootlim$ such that $\ket v = \rootlim \ket{w}$
		with $\lnode[w]{\id}{v_0}{\highlim}{v_1}$.
		The LIM $\highlim$ is chosen canonically as the lexicographically smallest.
		It runs in $O(n^3)$-time (with $n$ the number of qubits),
		provided $\getautomorphisms$ has been computed for children $v_0, v_1$
		(an amortized cost).
	}
    \label{alg:find-canonical-edges}
	\begin{algorithmic}[1]
		\Procedure{GetLabels}{\textsc{PauliLim} $\lambda P$, \Node $v_0, v_1$ \textbf{with} $\lambda \neq 0$ and $v_0, v_1$ reduced}
		\Statex \textbf{Output}: canonical high label $\highlim$ and root label $\rootlim$
		\State $G_0, G_1 := \getautomorphisms(v_0), \getautomorphisms(v_1)$
		\State $(g_0, g_1) := \textsc{ArgLexMin}(G_0, G_1, \lambda P)$
		\label{line:getlabels-argmin}
		\If{$v_0=v_1$}
		\label{algline:getlabels-start-minimizing}
		\State $(x,s):=\displaystyle\argmin_{(x,s)\in\{0,1\}^2} \set{(-1)^s\lambda^{(-1)^x}g_0Pg_1}$
		\Else
		\State $x:=0$
		\State $s:=\displaystyle \argmin_{s\in \{0,1\}} \set{ (-1)^s\lambda g_0Pg_1 }$
		\label{line:minimized-lim} 
		\EndIf
		\State $\highlim := (-1)^s \cdot \lambda^{(-1)^x} \cdot g_0 \cdot P \cdot g_1$
		\State $\rootlim := (X \otimes \lambda P)^{x} \cdot (Z^s \otimes (g_0)^{-1})$
		\State \Return $(\highlim, \rootlim)$
		\EndProcedure
	\end{algorithmic}
    \label{alg:get-labels}
\end{algorithm}

\subsubsection{Measurements}
\label{subsec:appendix_measurements_efficiency}

Now we show that measurements on top-qubits can be applied efficiently, i.e., with linear-sized representable numbers. We give a characterization for all numbers encountered during the measurement algorithm in \autoref{lemma:squaredNorm_analysis}. Then, in \autoref{lemma:SquaredNorm_values_are_efficiently_representable}, we show that these numbers are all in $R_{O(n+t)}$. Finally, all ingredients for the proof are collected in \autoref{cor:measurements_with_efficient_numbers}.

\begin{lemma}
\label{lemma:squaredNorm_analysis}
    \sloppy The algorithm \textsc{SquaredNorm} called on node $v$ reachable by path $x_1\dots,x_k$ returns the value $\sum_{x_{k+1}\dots x_n\in\{0,1\}^{n-k-1}}\prod_{d=k}^n |\alpha_{x_1\dots x_d}|^2$.
\end{lemma}
\begin{proof}
    The proof goes by induction.

    In the case that $k=n$ (i.e., $\index(v)=0$), the function \textsc{SquaredNorm} just returns the square normed edge label $|\alpha_{x_1\dots x_n}|^2$.

    \sloppy For general $k$, we use induction for the summation in line 8. The \textsc{SquaredNorm} algorithm is called on nodes reachable by paths $x_1\dots,x_k0$ and $x_1\dots,x_k1$, returning values $\sum_{x_{k+2}\dots x_n\in\{0,1\}^{n-k-2}}\prod_{d=k+1}^n |\alpha_{x_1\dots x_k 0\dots x_n}|^2$ and $\sum_{x_{k+2}\dots x_d\in\{0,1\}^{n-k-2}}\prod_{d=k+1}^n |\alpha_{x_1\dots x_k 1\dots x_d}|^2$. Summing them in line 8 yields $\sum_{x_{k+1}\dots x_n\in\{0,1\}^{n-k-1}}\prod_{d=k+1}^n |\alpha_{x_1\dots x_d}|^2$. The multiplication with $\lambda = \alpha_{x_1\dots x_k}$ in line 9 yields the returned value $\sum_{x_{k+1}\dots x_n\in\{0,1\}^{n-k-1}}\prod_{d=k}^n |\alpha_{x_1\dots x_d}|^2$.
\end{proof}

\begin{lemma}
\label{lemma:seq_edge_label_prod_is_fraction}
    The value $\prod_{d=k+1}^n |\alpha_{x_1\dots x_d}|^2$ equals $\left|\frac{\langle x_1\dots x_n\mid\phi\rangle}{\langle x_1\dots x_k 0\dots 0\mid\phi\rangle }\right|^2=\frac{\langle x_1\dots x_n\dyad{\phi} x_1\dots x_k 0\dots 0 \rangle}{\langle x_1\dots x_k 0\dots 0\dyad{\phi} x_1\dots x_k 0\dots 0\rangle}$.
\end{lemma}
\begin{proof}
    The proof follows by the same reasoning as in \autoref{lem:DDlabel=fraction-of-C_k} and \autoref{fig:quotient-of-statevectorentries}, with the extension that a multiplication over multiple edge values is used instead of a single edge value $a$.
\end{proof}

\begin{lemma}
\label{lemma:SquaredNorm_values_are_efficiently_representable}
    For any path $x_1\dots x_k$ the value $\sum_{x_{k+1}\dots x_n\in\{0,1\}^{n-k-1}}\prod_{d=k}^n |\alpha_{x_1\dots x_d}|^2$ is in $R_{O(n+t)}$.
\end{lemma}

\begin{proof}
    \sloppy By \autoref{lemma:seq_edge_label_prod_is_fraction} we know that we can write $\sum_{x_{k+1}\dots x_n\in\{0,1\}^{n-k-1}}\prod_{d=k}^n |\alpha_{x_1\dots x_d}|^2 = \frac{1}{\langle x_1\dots x_k 0\dots 0\dyad{\phi} x_1\dots x_k 0\dots 0\rangle}\sum_{x_{k+1}\dots x_n\in\{0,1\}^{n-k-1}}\langle x_1\dots x_n\dyad{\phi} x_1\dots x_k 0\dots 0 \rangle$. By \autoref{lemma:vector-entry-lemma} we know that $\langle x_1\dots x_n\dyad{\phi} x_1\dots x_k 0\dots 0 \rangle$ is in $Q_{n,t}$. A sum of at most $2^n$ of numbers in $Q_{n,t}$ is in $Q_{2n,t}$. As we know again by \autoref{lemma:vector-entry-lemma} that $\langle x_1\dots x_k 0\dots 0\dyad{\phi} x_1\dots x_k 0\dots 0\rangle$ is also in $Q_{n,t}$, and that this number is real, we can use \autoref{lemma:ratio_of_state-vector_entries_is_linearly_representable}, showing that the value $\sum_{x_{k+1}\dots x_n\in\{0,1\}^{n-k-1}}\prod_{d=k}^n |\alpha_{x_1\dots x_d}|^2$ is in $R_{O(n+t)}$.
\end{proof}

\begin{corollary}
\label{cor:measurements_with_efficient_numbers}
    All values encountered during the execution of a measurement (using the algorithms from \autoref{alg:measurement-top-qubit}) are in $R_{O(n+t)}$.
\end{corollary}
\begin{proof}
    By combining  \autoref{lemma:squaredNorm_analysis} and \autoref{lemma:SquaredNorm_values_are_efficiently_representable} we know that all numbers encountered during the execution of the \textsc{SquaredNorm} algorithm are representable in $R_{O(n+k)}$. As the \textsc{MeasurementProbability} algorithm does only a finite number of arithmetic operations on the numbers returned by \textsc{SquaredNorm}, we know by \autoref{lemma:basic_operations_in_qisq_are_tractable} that all numbers encountered during the execution of the \textsc{MeasurementProbability} algorithm are in $R_{O(n+k)}$.
\end{proof}

\begin{algorithm}[t!]
	\begin{algorithmic}[1]
		\Procedure{MeasurementProbability}{\Edge $e$}
		\State $s_0 := \textsc{SquaredNorm}(\follow 0e)$
		\State $s_1 := \textsc{SquaredNorm}(\follow 1e)$
		\State \Return $s_0/(s_0+s_1)$
		\EndProcedure
        \Procedure{SquaredNorm}{$\Edge \ledge{\lambda P}{v}$ \textbf{with} $\lambda \in \mathbb C, P\in \Pauli^{\index(v)}$}
		\If{$\index(v)=0$}
		\Return $|\lambda|^2$
		\EndIf
		\If{$v \notin \textsc{SNorm}\cache$}\label{sn:cache1} \Comment{Compute result once for $v$ and store in cache:}
		\State $\textsc{SNorm}\cache[v] :=  \textsc{SquaredNorm}(\follow 0{\ledge {\mathbb I}v}) + \textsc{SquaredNorm}(\follow 1{\ledge {\mathbb I}v})$\label{sn:cache2}
		\EndIf
		\State \Return $|\lambda|^2 \cdot \textsc{SNorm}\cache[v]$
							\Comment{Retrieve result for $v$ from cache and multiply with $|\lambda|^2$}
		\EndProcedure
	\end{algorithmic}
    \caption{(Reproduced from~\cite{vinkhuijzen2023limdd}) Algorithms \textsc{MeasurementProbability} and for computing the probability of observing outcome $\ket 0$ when measuring the top qubit of a Pauli \limdd in the computational basis.
		The subroutine \textsc{SquaredNorm} takes as input a Pauli \limdd edge $e$, and returns $\braket{e|e}$.
		It uses a cache to store the value $s$ of a node $v$.}
	\label{alg:measurement-top-qubit}
\end{algorithm}

\section{Clifford+$T$ simulation with decision diagrams is tractable in $T$-count \label{appendix:numnodes_bounds}}

In this appendix section, we provide some proofs and extensions from lemmas and theorems in \Cref{sec:numnodes_bound}. %

\subsection{LIMDD width is tractable in $T$-count}

\begin{lemma}%
\label{lemma:CliffTstabilizers}
    Every quantum state that can be reached from the $\ket{0}^{\otimes n}$ state with Clifford gates and $t$ $T$-gates ($R_z$ rotation gates) has at least $n-t$ linearly independent Pauli string stabilizers.
\end{lemma}
\begin{proof}
    We do this proof with induction to $t$.

    For $t=0$ the lemma is direct as every $n$-qubit stabilizer state has a stabilizer tableau with $n$ rows (see e.g.~\cite{aaronson2004improved}).
    
    Consider a quantum state with $m$ linearly independent Pauli string stabilizers. So the stabilizer tableau of this state has $m$ rows. Assume we want to apply a $T$-gate (or $R_z$ rotation gate) to the $k$'th qubit. We can bring the stabilizer tableau in reduced-row echelon form, as explained by~\cite{nonuni}. Then, without loss of generality, there is at most one tableau row with $X$ or $Y$ at qubit $k$, and the $m-1$ others have $Z$ or $I$ at qubit $k$. Conjugation of the $T$-gate (and the $R_z$ rotation gate) with $Z$ and $I$ does not change the $Z$ resp. $I$. Hence, all these $m-1$ stabilizer tableau rows are stabilizers of the resulting state. This shows the induction hypothesis.
\end{proof}

The following lemma is a variation of Lemma~\ref{lemma:CliffTstabilizers} for Toffoli gates instead of $T$ gates. This gives a stronger bound when Toffoli gates are considered instead of $T$ gates. %
\begin{lemma}
\label{lemma:CliffToffolistabilizers}
    Applying a Toffoli gate to quantum state $\ket{\psi}$ with $k$ linearly independent Pauli string stabilizers gives a resulting quantum state $\ket{\phi}$ with at least $k-3$ linearly independent Pauli string stabilizers.
\end{lemma}
\begin{proof}
    
    Consider all the $k$ stabilizers of $\ket{\psi}$. We can bring these stabilizers (i.e., the stabilizer tableau) in reduced-row echelon form, as explained in~\cite{nonuni}. Without loss of generality, we can assume that the control qubits ($C_1$,$C_2$) and the target qubit ($Tar$) of the Toffoli are the first three qubits of the stabilizer tableau. Hence, we have the following stabilizers that have no identities on the first three qubits:

    \begin{table}[!h]
        \centering
        \begin{tabular}{lcccccccc}
               && $C_1$ && $C_2$ && $Tar$ \\
            (1)&& X & $\otimes$ & . & $\otimes$ & . &$\otimes$ &\dots\\
            (2)&& Z & $\otimes$ & . & $\otimes$ & . &$\otimes$ &\dots\\
            (3)&& I & $\otimes$ & X & $\otimes$ & . &$\otimes$ &\dots\\
            (4)&& I & $\otimes$ & Z & $\otimes$ & . &$\otimes$ &\dots\\
            (5)&& I & $\otimes$ & I & $\otimes$ & X &$\otimes$ &\dots\\
            (6)&& I & $\otimes$ & I & $\otimes$ & Z &$\otimes$ &\dots
        \end{tabular}
        \caption{}
        \label{tab:stabilizers_Toffoli}
    \end{table}
    Here, the dots are undefined stabilizers $I,X,Z$. We will show that under conjugation with the Toffoli gate at most three of these Pauli string stabilizers can be converted into a non-Pauli string.

    Note that $Z\otimes I\otimes I$ and $I\otimes Z\otimes I$ and $I\otimes I\otimes X$ are invariant under conjugation with Toffoli. These Pauli products span half of the stabilizers in Table~\ref{tab:stabilizers_Toffoli}. Thus at most half of the stabilizers (1) to (6) are converted into non-Pauli strings under conjugation with Toffoli. 

    So we see that at most three Pauli string stabilizers can be converted into a non-Pauli string by conjugation with Toffoli.
    Hence, we conclude that the resulting quantum state $\ket{\phi}$ must have at least $k-3$ linearly independent Pauli string stabilizers.
\end{proof}

\begin{note}
\label{obs:stabilizer}
Let $\ket{v} = \ket{0}\ket{v_0} + \ket{1}\ket{v_1}$ be a quantum state. Let $P$ be a Pauli string. Assume that resp. $I\otimes P,X\otimes P,Y\otimes P$ or $Z\otimes P$ are a stabilizer of $\ket{v}$. Then we have:
\begin{itemize}
    \item[($IP$)] \begin{equation}
    I\otimes P\ket{v} = \ket{0}P\ket{v_0} + \ket{1}P\ket{v_1} = \ket{0}\ket{v_0} + \ket{1}\ket{v_1}
    \end{equation}
    Hence $P\ket{v_0} = \ket{v_0}$ and $P\ket{v_1} = \ket{v_1}$. So $P$ is a stabilizer of $\ket{v_0}$ and $\ket{v_1}$.

    \item[($XP$)] \begin{equation}
    X\otimes P\ket{v} = \ket{1}P\ket{v_0} + \ket{0}P\ket{v_1} = \ket{0}\ket{v_0} + \ket{1}\ket{v_1}
    \end{equation}
    Hence $P\ket{v_0} = \ket{v_1}$ and $P\ket{v_1} = \ket{v_0}$. So $\ket{v_0}$ and $\ket{v_1}$ are Pauli equivalent.

    \item[($YP$)] \begin{equation}
    Y\otimes P\ket{v} = i\ket{1}P\ket{v_0} -i \ket{0}P\ket{v_1} = \ket{0}\ket{v_0} + \ket{1}\ket{v_1}
    \end{equation}
    Hence $iP\ket{v_0} = \ket{v_1}$ and $-iP\ket{v_1} = \ket{v_0}$. So $\ket{v_0}$ and $\ket{v_1}$ are Pauli equivalent.

    \item[($ZP$)] \begin{equation}
    Z\otimes P\ket{v} = \ket{0}P\ket{v_0} - \ket{1}P\ket{v_1} = \ket{0}\ket{v_0} + \ket{1}\ket{v_1}
    \end{equation}
    Hence $P\ket{v_0} = \ket{v_0}$ and $-P\ket{v_1} = \ket{v_1}$. So $P$ is a stabilizer of $\ket{v_0}$ and $-P$ is a stabilizer of $\ket{v_1}$.
\end{itemize}
\end{note}

\begin{lemma}
\label{lemma:splitCase}
    Let $\ket{v} = \ket{0}\ket{v_0} + \ket{1}\ket{v_1}$ be a LIMDD node with $k$ linearly independent Pauli string stabilizers $P_1,\dots,P_k$. Then (at least) one of the following hold:\\
    (1) $\ket{v_0} \simeq_{Pauli} \ket{v_1}$ and $\ket{v_0}$ has at least $k-1$ linearly independent Pauli string stabilizers.\\
    (2) $\ket{v_0}$ and $\ket{v_1}$ have at least $k$ linearly independent Pauli string stabilizers.
\end{lemma}
\begin{proof}
    We start with the observation that all $P_i$ are of the form $P^{(1)}_i\otimes P^{(2\dots n)}_i$. Note that we can multiply the Pauli string stabilizers of $\ket{v}$ with each other, so we assume without loss of generality that $P^{(1)}_1\in\{I,X,Y,Z\}$ and $P^{(1)}_2\in\{I,Z\}$ and $P^{(1)}_i=I$ for all $i\in\{3,\dots,k\}$.\\

    If $P^{(1)}_1\in\{X,Y\}$, we see by \Cref{obs:stabilizer} that $\ket{v_0}$ and $\ket{v_1}$ are Pauli equivalent. Moreover, by \Cref{obs:stabilizer} we observe that $P^{(2\dots n)}_i$ for $i\in\{2,\dots,k\}$ are Pauli string stabilizers of $\ket{v_0}$ as $P^{(1)}_i\in\{I,Z\}$ for $i\geq2$. Thus we are in case (1).\\

    If $P^{(1)}_1\in\{I,Z\}$, we see by \Cref{obs:stabilizer} that $P^{(2\dots n)}_i$ for $i\in\{1,\dots,k\}$ are Pauli string stabilizers of $\ket{v_0}$. Similarly, $\pm P^{(2\dots n)}_i$ for $i\in\{1,\dots,k\}$ are Pauli string stabilizers of $\ket{v_1}$, where the sign of the stabilizer is $+$ if $P^{(1)}_i=I$ and $-$ if $P^{(1)}_i=Z$. Thus we are in case (2).
\end{proof}

\begin{theorem}%
\label{thm:m_stabs_LIMDD_width}
    Every $n$ qubit quantum state with $m\leq n$ linearly independent Pauli string stabilizers has a LIMDD representation with width at most $2^{n-m}$.
\end{theorem}
\begin{proof}
    This proof goes with induction to $n$.\\
    For $n = 1$, this theorem is trivially true for every $m\leq n$.\\
    Assume this theorem is true for all $n<n_0$. We will show that the theorem holds for $n_0$.
    
    Consider any state on $n_0$ qubits with $m\leq n_0$ linearly independent Pauli stabilizers. Note that for $m=n_0$ this theorem is true, as every $n_0$ qubit states with $n_0$ linearly independent Pauli string stabilizers is a stabilizer state \cite[Theorem 1]{aaronson2004improved}, and every stabilizer state can be represented as a Tower-LIMDD, which has width 1~\cite{vinkhuijzen2023limdd}. Thus we assume $m<n_0$. 
    \\Pick a LIMDD representation of this state with minimal LIMDD-width. Consider the top node $\ket{v} = \ket{0}\ket{v_0} + \ket{1}\ket{v_1}$ of this LIMDD. By Lemma~\ref{lemma:splitCase}, we can have two cases for $\ket{v}$.
    
    If (1) holds, we have $\ket{v_0}\simeq_{Pauli}\ket{v_1}$, and hence the LIMDD-width of $\ket{v}$ equals the LIMDD-width of $\ket{v_0}$ as the LIMDD does not split at the top layer. Note that $\ket{v_0}$ is a $n_0-1$ qubit state and by Lemma~\ref{lemma:splitCase} it has at least $m-1$ linearly independent Pauli string stabilizers. So by the induction hypothesis, we know that $\ket{v_0}$ has LIMDD-width at most $2^{(n_0-1)-(m-1)} = 2^{n_0-m}$, which shows the induction hypothesis.

    If (2) holds, we know by Lemma~\ref{lemma:splitCase} that $\ket{v_0}$ and $\ket{v_1}$ have at least $m$ linearly independent Pauli string stabilizers. As $\ket{v_0}$ and $\ket{v_1}$ are $n_0-1$ qubit states, we know by the induction hypothesis that $\ket{v_0}$ and $\ket{v_1}$ have LIMDD-width at most $2^{(n_0-1)-m}$. %
    Hence, as $\ket{v}$ has only two child nodes $\ket{v_0}$ and $\ket{v_1}$, $\ket{v}$ has LIMDD-width at most $2^{(n_0-1)-m}+2^{(n_0-1)-m} = 2^{n_0-m}$, which shows the induction hypothesis.\\
    This concludes the induction step for the $n_0$ qubit state case.    
\end{proof}

\subsection{EVDD width is tractable in H,T and controlled-Pauli count}

\begin{lemma}
    Every $n$ qubit quantum state that can be reached from the $\ket{0}^{\otimes n}$ state with $X,Y,Z,S,T,H,CZ,SWAP$ gates has at least $n-m$ linearly independent Pauli string stabilizers of the form $\pm Z_k$, where $m=\#H$.
    \label{lemma:h-count_EVDD}
\end{lemma}
\begin{proof}
    Note that the $\ket{0}$ state has stabilizers $Z_1,\dots, Z_n$. So there are $n$ stabilizers of the form $Z_k$.
    
    Applying a gate $G$ will turn stabilizer $P$ into stabilizer $GPG^\dagger$. 
    
    For $G=X,Y,Z,S,T$ we know that $GZG^\dagger = \pm Z$. Moreover, $$CZ_{i,j}(Z\otimes I)CZ_{i,j}^\dagger = (Z\otimes I) \quad \text{and} \quad CZ_{i,j}(I\otimes Z)CZ_{i,j}^\dagger = (I\otimes Z).$$ Also, $$SWAP_{i,j}(Z\otimes I)SWAP_{i,j}^\dagger = (I\otimes Z) \quad \text{and} \quad SWAP_{i,j}(I\otimes Z)SWAP_{i,j}^\dagger = (Z\otimes I).$$
    Thus, any gate $X,Y,Z,S,T,CZ,SWAP$ does not reduce the number of stabilizers of the form $\pm Z_k$.

    For $G=H$ we have $HZH^\dagger = X$. Thus, an $H$ gate can reduce the number of $\pm Z_k$ stabilizers by 1. Hence, applying $m$ $H$-gates can reduce the number of $\pm Z_k$ by $m$. 
    
    So there are $n-m$ stabilizers of the form $\pm Z_k$.
\end{proof}

\begin{lemma}%
    Every $n$ qubit quantum state that can be reached from the $\ket{0}^{\otimes n}$ state with $X,Y,Z,S,T,H,CZ,SWAP$ gates has at least $n-m$ linearly independent Pauli string stabilizers of the form $\pm P_k$, where $m=\min{(\#H,2\cdot \#CZ + \#T)}$.
    \label{lemma:h-cz-t-count_EVDD}
\end{lemma}
\begin{proof}
    From \Cref{lemma:h-count_EVDD} we know that after applying $\#H$ $H$-gates there are $n-\#H$ stabilizers of the form $\pm Z_k$. There are also $\#H$ stabilizers of the form $\pm X_k$.%

    The gates $X,Y,Z,S,SWAP$ do not change the number of stabilizers of the form $\pm X_k$ and $\pm Y_k$ as they are invariant under conjugation.

    Conjugation of $\pm X_k$ or $\pm Y_k$ with the $T$-gate can turn such a stabilizer into a non-Pauli stabilizer. E.g. $TXT^\dagger = \frac{X+Y}{\sqrt{2}}$.

    For the $CZ_{i,j}$ gate we have $$CZ_{i,j}(I\otimes X)CZ_{i,j}^\dagger = Z\otimes X \quad \text{and}\quad CZ_{i,j}(I\otimes Y)CZ_{i,j}^\dagger = Z\otimes Y.$$
    In the case that $\pm Z_i$ is a stabilizer, we see that $\pm I\otimes X$ resp. $\pm I\otimes Y$ is again a local stabilizer (because stabilizers can be multiplied). \\
    If $\pm Z_i$ is not a stabilizer, then $Z\otimes X$ resp. $Z\otimes Y$ is not a local stabilizer. But then at least two $H$-gates needed to be applied to have this case. So an $\pm X_k$ or $\pm Y_k$stabilizer can only be turned into a stabilizer that is not of the form $\pm P_k$ if two $H$-gates were applied to the corresponding qubits (modulo $SWAP$ gates).\\
    Similar arguments holds for $CZ_{i,j}(X\otimes I)CZ_{i,j}^\dagger = X\otimes Z$. Therefore, a $CZ$ gate can `destroy' two $\pm P_k$ stabilizers, but only if two $H$-gates were applied to the corresponding qubits (modulo $SWAP$ gates).

    So there are indeed $n-m$ stabilizers of the form $\pm P_k$ where $m=\min{(\#H,2\cdot \#CZ + \#T)}$.
\end{proof}

\begin{theorem}%
    Every $n$ qubit quantum state with $m\leq n$ linearly independent Pauli string stabilizers of the form $\pm P_k$ has a EVDD representation with width at most $2^{n-m}$.
    \label{theorem:EVDD_width_from_number_of_stabilizers}
\end{theorem}
\begin{proof}
    The main idea of the proof is that for every stabilizer $\pm P_k$ the EVDD nodes at level $k$ have at most one output node. Therefore, the number of nodes cannot increase between level $k$ and $k+1$. As the number of child nodes is at most two for every node, the number of nodes can at most double between layers $l$ and $l+1$ without a stabilizer $\pm P_l$. Thus, the bound on the EVDD width is $2^{n-m}$.
    
    Now, we will formalize this as follows.
    
    Every node in the EVDD at level $k$ represents a vector $$\ket{v^{(k)}} = \alpha_0\ket{0}\ket{v^{(k+1)}_0} + \alpha_1\ket{1}\ket{v^{(k+1)}_1}.$$ Here $\ket{v^{(k+1)}_0}$ and $\ket{v^{(k+1)}_1}$ represent the child nodes of $\ket{v^{(k)}}$ and $\alpha_0$ and $\alpha_1$ the edge labels.

    If $\pm P_k$ is a stabilizer, then we have $\ket{v^{(k)}} = \pm P_k \ket{v^{(k)}}$. Hence, $$\alpha_0\ket{0}\ket{v^{(k+1)}_0} + \alpha_1\ket{1}\ket{v^{(k+1)}_1} = \pm \alpha_0 P\ket{0}\ket{v^{(k+1)}_0} \pm \alpha_1 P\ket{1}\ket{v^{(k+1)}_1}.$$

    For $P = X$ we thus have $\ket{v^{(k+1)}_0} = \pm \frac{\alpha_1}{\alpha_0}\ket{v^{(k+1)}_1}$.\\
    For $P = Y$ we thus have $\ket{v^{(k+1)}_0} = \pm i \frac{\alpha_1}{\alpha_0}\ket{v^{(k+1)}_1}$.\\
    For $P = Z$ we have either $\ket{v^{(k+1)}_0} = -\ket{v^{(k+1)}_0}$ or $\ket{v^{(k+1)}_1} = - \ket{v^{(k+1)}_1}$, implying respectively $\ket{v^{(k+1)}_0} = 0$ and $\ket{v^{(k+1)}_1} = 0$.

    Thus in all cases the child nodes $\ket{v^{(k+1)}_0}$ and $\ket{v^{(k+1)}_1}$ are the same up to multiplication with a complex scalar. Hence, these child nodes can be merged in the EVDD. Thus in all cases every node at level $k$ has at most one child node if $\pm P_k$ is a stabilizer.
\end{proof}

\end{document}